\documentclass[10pt,thm-restate]{article}

\usepackage{pgfplots}
\pgfplotsset{compat=1.14}
\usepackage{tikz-3dplot}
\usepackage{xspace}
\usepackage{tikz}
\usepackage{multirow}
\usepackage{marvosym}
\usepackage{enumitem}
\usepackage{fullpage}
\usepackage{hyperref}
\usepackage{amsthm}
\usepackage{amssymb}
\usepackage{mathtools}
\usepackage{amsfonts}
\usepackage{booktabs}
\usepackage{subcaption}
\usepackage{relsize}
\usepackage[T1]{fontenc}

\usepackage{amsmath}
\usepackage{hyperref}
\usepackage{pgfplots}
\pgfplotsset{compat=1.16}
\usepackage{tikz-3dplot}
\usepackage{xspace}
\usepackage{tikz}
\pgfdeclarelayer{background}
\pgfsetlayers{background,main}

\usepackage{float}
\usepackage{multirow}
\usepackage{marvosym}
\usepackage{enumitem}
\usepackage{booktabs}
\usepackage{subcaption}
\usepackage{tcolorbox}
\usepackage{mdframed}
\usepackage{minibox}
\usepackage{xcolor,colortbl}

\usepackage{thmtools, thm-restate}

\usepackage{bm} 
\usepackage{url}
\usepackage{amsmath}
\usepackage{amsfonts}
\usepackage{amsthm}
\usepackage{algorithmic}
\usepackage{graphicx}
\usepackage{textcomp}
\usepackage{xcolor}
\usepackage{enumerate}   
\def\BibTeX{{\rm B\kern-.05em{\sc i\kern-.025em b}\kern-.08em
    T\kern-.1667em\lower.7ex\hbox{E}\kern-.125emX}}

\usepackage{hyperref}
\usepackage[colorinlistoftodos]{todonotes}
\usepackage[ruled, noend]{algorithm2e}
\usepackage{tikz}
\usepackage{microtype}

\newtheoremstyle{cited}%
{.5\baselineskip\@plus.2\baselineskip
    \@minus.2\baselineskip}
{.5\baselineskip\@plus.2\baselineskip
    \@minus.2\baselineskip}
{\itshape}
{\parindent}
{}
{.}
{.5em}
{\textsc{\thmname{#1}} \thmnote{\normalfont#3}}

\theoremstyle{cited}

\newcounter{magicrownumbers}

\theoremstyle{plain}                  
\newtheorem{theorem}{Theorem}
\newtheorem{lemma}[theorem]{Lemma}
\newtheorem{proposition}[theorem]{Proposition}

\newtheorem{corollary}[theorem]{Corollary}         
\newtheorem{definition}[theorem]{Definition}
\newtheorem{example}[theorem]{Example}

\newtheoremstyle{cited}%
{.5\baselineskip\@plus.2\baselineskip
    \@minus.2\baselineskip}
{.5\baselineskip\@plus.2\baselineskip
    \@minus.2\baselineskip}
{\itshape}
{\parindent}
{}
{.}
{.5em}
{\textsc{\thmname{#1}} \thmnote{\normalfont#3}}

\newcommand{\set}[1]{\{#1\}}                    
\newcommand{\setof}[2]{\{{#1}\mid{#2}\}}        

\newcommand{\dan}[1]{\todo[inline,color=yellow]{\textsf{#1} \hfill \textsc{--Dan.}}}

\newcommand{\calC}{\mathcal C}
\newcommand{\calG}{\mathcal G}

\newtheorem{thm}{Theorem}[section]
\newtheorem{lmm}[thm]{Lemma}

\newtheorem{claim}[thm]{Claim}

\newcommand{\defeq}{\stackrel{\text{def}}{=}}

\newcommand{\shap}{\texttt{Shap}}

\newcommand{\TAB}{\makebox[2.5ex][r]{}}%

\newcommand{\nop}[1]{}

\date{}
\title{From Shapley Value to Model Counting and Back}

\author{
\begin{tabular}{cccc}
Ahmet Kara$^1$  & Dan Olteanu$^1$ &  Dan Suciu$^2$ \\
kara@ifi.uzh.ch &  olteanu@ifi.uzh.ch &  suciu@cs.washington.edu 
\end{tabular}\\ \\
$^1$University of Zurich  \enspace\enspace $^2$University of Washington
}

\begin{document}
\maketitle

\begin{abstract}

In this paper we investigate the problem of quantifying the contribution of each variable to the satisfying assignments of a Boolean function based on the Shapley value.

Our main result is a polynomial-time equivalence between computing Shapley values  and model counting for any class of Boolean functions that are closed under substitutions of variables with disjunctions of fresh variables. This result settles an open problem raised in prior work, which sought to connect the Shapley value computation to probabilistic query evaluation.

We show two applications of our result. First, the Shapley values can be computed in polynomial time over deterministic and decomposable circuits, since they are closed under OR-substitutions. Second, there is a polynomial-time equivalence between computing the Shapley value for the tuples contributing to the answer of a Boolean conjunctive query and counting the models in the lineage of the query. This equivalence allows us to immediately recover the dichotomy for Shapley value computation in case of self-join-free Boolean conjunctive queries; in particular, the hardness for non-hierarchical queries can now be shown  using a simple reduction from the \#P-hard problem of model counting for lineage in positive bipartite disjunctive normal form.
\end{abstract}

\section{Introduction}
\label{sec:introduction}
The Shapley value quantifies the fair contribution of a player to a wealth function that is shared by a set of players in a cooperative game~\cite{Shapley1953,roth1988shapley}. \nop{It is a principled  approach to assign fair scores to the players~\cite{roth1988shapley}.} For this reason, it has been used in a variety of applications ranging from bioinformatics to network analysis and machine learning: measuring the centrality and power of genes~\cite{MorettiFPB10} and the influence in social networks~\cite{NarayanamN11}; sharing profit between Internet providers~\cite{MaCLMR08,MaCLMR10}; finding key players in networks~\cite{CampenHHL18}; feature selection, explainability, multi-agent reinforcement learning, ensemble pruning, and data valuation~\cite{RozemberczkiWBY:ShapleyML:2022,MinhWLN:ExplainableAI:2022}. 

In this paper we investigate the problem of computing the Shapley value for  variables in Boolean functions. The Shapley values quantify the contribution of each variable to the satisfying assignments of the Boolean function.
Understanding the importance of variables to the outcome of a Boolean function has numerous applications~\cite{HammerKR00, HarderJBD20}. The nature of the Shapley values for the variables in Boolean functions can also serve as complexity-theoretic assumption for tractability in generalized constraint satisfaction problems with order predicates~\cite{BrakensiekGS23, kalai2004social}.
When focusing on functions representing the lineage of Boolean conjunctive queries in relational databases~\cite{DBLP:conf/pods/GreenKT07}, the Shapley values are used to support explanations for query answers. In this setting, the tuples in the input database are the players that contribute to the answer of a given query and the Shapley value assigns a score to each input tuple based on its contribution to the query answer.
Recent works in database theory and systems~\cite{DavidsonDFKKM22,DeutchFKM22,LivshitsBKS21,ReshefKL20} have made great progress towards charting the tractability frontier of computing the Shapley values of database tuples and proposed algorithms for exact and approximate computation.
We next highlight two key results from prior work.

First, for every Boolean query $Q$ and database $\bm D$, the problem of computing the Shapley value of any tuple in $\bm D$ reduces in polynomial time to the problem of computing $Q$ over a probabilistic version of $\bm D$, where each tuple becomes an independent random variable~\cite{DeutchFKM22}. This connection to probabilistic query evaluation (PQE) allows to transfer well-established results from PQE to Shapley value computation. In particular, the tractability of PQE for safe queries~\cite{Suciu:PDB:2011} implies the tractability of Shapley value computation for safe queries. Furthermore, knowledge compilation techniques developed for PQE can be adjusted for Shapley value computation. It is stated as open problem whether PQE also reduces in polynomial time to Shapley value computation, effectively establishing a polynomial-time equivalence between the two problems~\cite{DeutchFKM22}.

Second, the dichotomy for conjunctive queries without self-joins over probabilistic databases~\cite{DalviS04} also holds for Shapley value computation~\cite{LivshitsBKS21}: For any self-join-free Boolean conjunctive query $Q$, the problem of Shapley value computation is in FP if $Q$ is hierarchical and is FP$^{\#\text{P}}$-hard otherwise.
\nop{It remains open whether the hardness of the intractable (unsafe) unions of conjunctive queries carries over from probabilistic databases~\cite{DalviS12} to Shapley values.}


The main result in this paper is a polynomial-time equivalence between the Shapley value computation and model counting for any class of Boolean functions that are closed under substitutions of variables with (possibly empty) disjunctions of fresh variables. This equivalence connects the Shapley value computation to a fundamental and well-established problem~\cite{GomesSS:ModelCounting:2021} with many applications from artificial intelligence to formal verification. This result settles the open problem raised in prior work~\cite{DeutchFKM22}, albeit not using PQE but model counting under OR-substitutions.

We also show two applications of our result. In Section~\ref{sec:circuits} we first show that deterministic and decomposable circuits are closed under OR-substitutions, where we allow further polynomial-time transformations. Since model counting is tractable for such circuits~\cite{DarwicheM:KCMap:2002}, it follows from our main result that Shapley value computation is also tractable for such circuits. Deterministic and decomposable circuits are extensively investigated in knowledge compilation~\cite{DarwicheM:KCMap:2002,DarwicheMSS:KnowledgeCompilation:2017}, prime examples are the ordered binary decision diagrams (OBDDs) and the deterministic decomposable negation normal forms (d-DNNFs).

Our second application is in databases. 
In Section~\ref{sec:queries} we show a polynomial-time equivalence between computing the Shapley value for the tuples contributing to the answer of a Boolean conjunctive query $Q$ and counting the models in the {\em lineage} of $Q$. When lifted to the level of the query, the OR-substitutions can be expressed by {\em stretching} the query, a rewriting which introduces fresh variables in relations.
This equivalence allows us to immediately recover the dichotomy for Shapley value computation in case of self-join-free Boolean conjunctive queries~\cite{LivshitsBKS21}; in particular, the hardness for non-hierarchical queries can now be shown  using a simple reduction from the \#P-hard problem of model counting for lineage in positive bipartite formulas in disjunctive normal form~\cite{ProvanB83}, as previously used to show FP$^{\text{\#P}}$-hardness of PQE~\cite{DalviS04}. 

\paragraph{Shapley value versus SHAP score} Recent works~\cite{BroeckLSS:AAAI:2021,BroeckLSS:SHAP:2022,ArenasBBM:AAAI:2021,ArenasBOS:Neurips:2022,arenas2023complexity} consider the notion of SHAP score, which is based on, yet different from, the Shapley value and used for providing explanations in machine learning. For a given classification model $M$, entity $\bm e$, and feature $x$, the SHAP score intuitively represents the importance of the feature value $\bm e(x)$  to the classification result $M(\bm e)$. In its general formulation, it takes as input a Boolean function $F$ encoding a Boolean classifier and a probability distribution on the set of truth assignments. The probability distribution is assumed to be a {\em product distribution}, also called a {\em fully factorized distribution}, and the wealth function of the SHAP score is an {\em expectation}. In this setting, it was shown that computing the SHAP score is polynomial-time equivalent to {\em weighted} model counting for the function $F$~\cite{BroeckLSS:AAAI:2021,BroeckLSS:SHAP:2022}. These prior works~\cite{ArenasBBM:AAAI:2021,ArenasBOS:Neurips:2022} also show that the SHAP score can be computed in polynomial time in case the Boolean function $F$ is given by a tractable (deterministic and decomposable) circuit. Tractability of such circuits is the main study in knowledge compilation~\cite{DarwicheM:KCMap:2002,DarwicheMSS:KnowledgeCompilation:2017}.


In contrast, we study the Shapley value where the wealth function is just the Boolean function $F$, without any probability distribution. This appears unrelated to the SHAP score, in particular it is not equivalent to setting all probabilities to 1/2. While there exist fully-polynomial randomized approximation schemes (FPRAS) for model counting~\cite{KarpLM:FPRAS:1989} and the Shapley value in the database context~\cite{LivshitsBKS21}, there is no such FPRAS for the SHAP score even in case of positive bipartite DNF functions~\cite{arenas2023complexity}.
Our polynomial-time equivalence is technically more challenging than for the SHAP score discussed in prior work~\cite{BroeckLSS:AAAI:2021,BroeckLSS:SHAP:2022}, because we no longer have the ability to use an oracle with varying probability functions (or, equivalently, weight functions).  Instead, our proof of equivalence relies on the ability to substitute a Boolean variable with a disjunction of fresh variables.

\section{Preliminaries}
\label{sec:prelims}
We use $\mathbb{N}$ to denote the set of natural numbers including $0$.
For $n\in\mathbb{N}$, we denote by $[n] \defeq \set{1,2,\ldots,n}$. In case $n=0$, then $[n]=\emptyset$.

\paragraph{Boolean Functions}
Let $\bm X$ be the set of $n \in \mathbb{N}$ Boolean variables $X_1,\ldots,X_n$. Where convenient, we may denote a variable $X_i$ by its index $i$.
A {\em Boolean function} over 
$n \in \mathbb{N}$ variables is a function 
$F:\set{0,1}^n \rightarrow \set{0,1}$ that uses the logical connectors $\wedge$ (and), $\vee$ (or), and $\neg$ (not). The {\em size} of a function $F$, denoted by $|F|$, is the number of occurrences of variables and of the logical connectors  in $F$.
We denote 
by $\bm{BF}$ the set of all Boolean functions.
For example, $F = X_1 \wedge (X_2 \vee \neg X_3)$ is a function 
over the three variables $X_1$, $X_2$, and $X_3$. 
We identify isomorphic functions, i.e., they are equal up to renaming of variables; $F$ can also be written as $Y_1 \wedge (Y_2\vee \neg Y_3)$.

\paragraph{Substitutions}
Given $n \in \mathbb{N}$,
a {\em substitution} is a function $\theta : [n] \rightarrow \bm{BF}$.
We often denote the substitution $\theta$ by the set $\set{X_1 := \theta(1), \ldots,$ $X_n := \theta(n)}$.
The result of applying the substitution $\theta$ to a Boolean function $F$
is denoted by $F[\theta]$.  
We may define a substitution only on a subset of the variables and assume implicitly that the other variables are mapped to themselves.  
For example, for the above function $F$ and substitution $\theta = \{X_2 := Z_1 \vee Z_2\}$, we have $F[\theta] = X_1 \wedge (Z_1 \vee Z_2 \vee \neg X_3)$.

\begin{definition} \label{def:stretch:c}
A Boolean function $F$ over $n$ variables {\em admits an OR-substitution into} a Boolean function $G$, denoted by  $F\overset{\text{OR}}{\rightarrow}G$, if 
$G = F[\theta]$ with $\theta = \set{X_{i}:= Z_i^1 \vee \ldots \vee Z_i^{m_i}|$ $i \in [n]}$ for 
$m_1, \ldots , m_n \in \mathbb{N}$ and fresh variables $Z_i^1,\ldots,Z_i^{m_i}$.
Notice that $G$ has $\sum_{i = 1}^n  m_i$ variables. If $m_i=0$, then $\theta$ maps $X_i$ to $0$.

Given a class $\calC$ of Boolean functions, we define $\widetilde{\calC} = \set{G | \exists F \in \calC \text{ such that } F\overset{\text{OR}}{\rightarrow}G}$.
We say that $\calC$ {\em OR-substitutes into} $\widetilde{\calC}$.
    \end{definition}

Notice that it holds $\calC \subseteq \widetilde{\calC}$, because we can
substitute each $X_i$ with a single variable $Z_i$ and obtain an isomorphic function.

\paragraph{Valuations}
{\em Valuations} are special substitutions where variables are mapped to constants. Given a valuation $\theta : [n] \rightarrow \set{0,1}$, we denote by $F[\theta]$ the Boolean value of $F$. We say that $\theta$ is a {\em model} of $F$ if $F[\theta] = 1$. It is often convenient to denote the valuation $\theta$ by the set $T \defeq \setof{i \in [n]}{\theta(i)=1}$, in which case we write $F[T]$ for $F[\theta]$. The size of a model $\theta$ is thus the number of variables it sets to 1, i.e., $|T|$. 
For instance, consider the valuation $T = \set{1}$, which for the example function $F = X_1 \wedge (X_2 \vee \neg X_3)$ maps $X_1$ to $1$ and the other two variables $X_2$ and $X_3$ to $0$. Then, $F[\set{X_1}] = 1$, so $T$ is a model of $F$ of size $1$.
Two functions $F_1$ and $F_2$ are {\em equivalent}, denoted by $F_1\equiv F_2$, if $F_1[\theta] = F_2[\theta]$ for all valuations $\theta$.

\paragraph{Model Counting}
Consider a Boolean function $F$ over $n$ variables.  The {\em model count}
$\#F$ is  the number of models of $F$:
\begin{align*}
  \#F \defeq & \sum_{T \subseteq [n]}F[T]
\end{align*}
Given $0\leq k \leq n$,  the {\em $k$-model count} $\#_k F$ is the number of models of $F$ of size $k$:
\begin{align*}
  \#_kF \defeq & \sum_{T \subseteq {[n] \choose k}} F[T]
\end{align*}
where ${[n] \choose k}$ represents the subsets of $[n]$ of size $k$.
We denote the vector of $k$-model counts by:
\begin{align*}
  \#_{0,\ldots,n}F \defeq & (\#_0F, \#_1F, \ldots, \#_nF)
\end{align*}

\paragraph{Shapley value}
Given a Boolean function $F$ over $n$ variables, the {\em Shapley value} of a variable $X_i$ for $i\in[n]$ is defined as:
\begin{align}
  \shap(F,X_i) \defeq & \frac{1}{n!} 
  \sum_{\Pi \in S_n} \left(F[\Pi^{<i} \cup \set{i}]-F[\Pi^{<i}]\right)\label{eq:shap}
\end{align}
where $S_n$ is the symmetric group, i.e., the set of permutations of $[n]$, and $\Pi^{<i}$ is the set of indices $j$ that come before $i$ in the permutation $\Pi$. 
If $i$ is at the first position of $\Pi$, then $\Pi^{<i}$ is the empty set.  
\begin{example}
\label{ex:model_counting_shapley}
Consider again the function 
$F = X_1 \wedge (X_2 \vee \neg X_3)$. The only models of the 
function are $\{X_1\}$, $\{X_1, X_2\}$, and $\{X_1, X_2, X_3\}$. Hence, 
$\#F = 3$, $\#_0F = 0$, and $\#_1F = \#_2F = \#_3F = 1$. The table below 
shows for each possible permutation $\Pi \in S_3$, the difference
$F[\Pi^{<i} \cup \set{i}]-F[\Pi^{<i}]$ for $i \in [3]$.
For instance, in case $\Pi= (2,1,3)$ we have
$\Pi^{<1} \cup \set{1} = \{1,2\}$ and $\Pi^{<1} = \{2\}$.
Hence, $F[\Pi^{<1} \cup \set{1}]-F[\Pi^{<1}] = 1-0 = 1$.
\begin{center}
\begin{tabular}{cccc}
 & \multicolumn{3}{c}{$F[\Pi^{<i} \cup \set{i}]-F[\Pi^{<i}]$}  \\
$\Pi$ & $i = 1$ & $i = 2$ & $i = 3$ \\ \hline
$(1,2,3)$ & $1$ & $0$ & $0$ \\
$(1,3,2)$ & $1$ & $1$ & $-1$ \\
$(2,1,3)$ & $1$ & $0$ & $0$ \\
$(2,3,1)$ & $1$ & $0$ & $0$ \\
$(3,1,2)$ & $0$ & $1$ & $0$ \\
$(3,2,1)$ & $1$ & $0$ & $0$ \\
\end{tabular}
\end{center}
To obtain the Shapley value of variable $X_i$, we sum up the values in the 
column for $i$ and divide by $3!=6$. We obtain 
$\shap(F,X_1) = \frac{5}{6}$, $\shap(F,X_2) = \frac{2}{6}$,
$\shap(F,X_3) = -\frac{1}{6}$. Note that the Shapley value 
of $X_3$ is negative because it appears negatively in the function.
\end{example}

Next, we give an alternative formulation of the Shapley value that uses model counting. 
%
\begin{proposition}[\cite{LivshitsBKS21} page 11, adapted]
\label{prop:alternative_shap}
The {\em Shapley value} of a variable $X_i$ of a Boolean function $F$ is:
\begin{align}
  \shap(F,X_i) = & \sum_{k=0}^{n-1} c_k \left(\#_kF[X_i:=1]-\#_kF[X_i:=0]\right) \label{eq:shap:c}
\end{align}
where $c_k = \frac{k! (n-k-1)!}{n!}$.
\end{proposition}
The above formulation does not consider $\#_nF$, since $X_i$ is set to either $1$ or $0$ and $F$ has therefore $n-1$ remaining variables.
\begin{example}
\label{ex:alternative_shapley}
We compute the Shapley value of $X_1$ in 
$F = X_1 \wedge (X_2 \vee \neg X_3)$
using Eq.~\eqref{eq:shap:c}.
We have $F[X_1:=0] = 0 \wedge (X_2 \vee \neg X_3)$.
Since this function cannot evaluate to $1$,
we have $\#_0F[X_1:=0] = \#_1F[X_1:=0] = \#_2F[X_1:=0] = 0$.
It holds $F[X_1:=1] = 1 \wedge (X_2 \vee \neg X_3)$.
The function 
$F[X_1:=1]$ has the models $\emptyset$, $\{X_2\}$, and 
$\{X_2,X_3\}$. 
Hence,
$\#_0F[X_1:=1] = \#_1F[X_1:=1] = \#_2F[X_1:=1] = 1$. 
We have $c_0 = \frac{0!(3-0-1)!}{6} = \frac{2}{6}$, 
$c_1 = \frac{1!(3-1-1)!}{6} = \frac{1}{6}$, and
$c_2 = \frac{2!(3-2-1)!}{6} = \frac{2}{6}$.
Following Eq.~\eqref{eq:shap:c}, we obtain 
$\shap(F,X_1) = \frac{2}{6} + \frac{1}{6} + \frac{2}{6}= \frac{5}{6}$, 
which is the
Shapley value of $X_1$ as computed in Example~\ref{ex:model_counting_shapley}.
\end{example}

The following proposition follows immediately from the definition of the Shapley value:
\begin{proposition}
\label{prop:sum_shap}
For any Boolean function $F$, it holds
\begin{align*}
  \sum_{i\in[n]}\shap(F,X_i) = & F[\bm 1] - F[\bm 0]
\end{align*}
where $\bm 1$ is the valuation that maps all variables to $1$, and
$\bm 0$ the valuation that maps all variables to $0$.
\end{proposition}
In the original setting, we have $\sum_{i\in[n]}\shap(F,X_i) = F[\bm 1]$. This does not hold in our case, since $F[\bm 0]$ may not necessarily be 0 as $F$ may have both positive and negative literals. That is, in our setting the {\em efficiency} property ($F[\bm 0]=0$) of the Shapley value~\cite{roth1988shapley} does not hold; it holds for functions where all literals are positive.

\begin{example}
\label{ex:sum_shapley}
For the function 
$F = X_1 \wedge (X_2 \vee \neg X_3)$,
we have $F[\bm 1] = 1$ and $F[\bm 0] =0$,
since $\bm 1$ is a model of $F$ but $\bm 0$
is not. By Proposition~\ref{prop:sum_shap}, the Shapley values of the variables of $F$ must sum up to $1$. This is indeed the case, since
we have $\shap(X_1) = \frac{5}{6}$, 
$\shap(X_2) = \frac{2}{6}$, and 
$\shap(X_3) = -\frac{1}{6}$
(see 
Example~\ref{ex:model_counting_shapley}). 
\end{example}

We write $\shap(F)$ to denote the  vector of the Shapley values of all variables in $F$:
\begin{align*}
  \shap(F) \defeq & (\shap(F,X_1), \ldots, \shap(F,X_n))
\end{align*}

\paragraph{Polynomial-time Reductions and Transformations}
A {\em polynomial-time reduction} (also called a {\em Cook reduction}) from a problem $A$ to a problem $B$,
denoted by $A~\leq^P~B$, is a
polynomial-time algorithm for the problem $A$ with access to an oracle
for the problem~$B$.  If $B~\leq^P~A$ also holds, then we write
$A \equiv^P B$ and say that the two problems are polynomial-time equivalent.
A {\em polynomial-time transformation} from a class of functions $\calC_1$ to another class of functions $\calC_2$, denoted by $\calC_1~\lesssim^P~\calC_2$, is an algorithm $T$ that takes time polynomial in the representation size of functions and such that: 
$\forall F_1\in\calC_1, \exists F_2\in\calC_2: F_2 = T(F_1) \text{ and } F_1$ $\equiv F_2$. If $\calC_2~\lesssim^P~\calC_1$ also holds, then we write $\calC_1~\approx^P~\calC_2$ and say that $\calC_1$ and $\calC_2$ have a bidirectional polynomial-time transformation.

\section{Polynomial-time Reductions for Problems over Boolean Functions}
\label{sec:main-results-functions}

We consider three problems: model counting, fixed-size model counting, and Shapley value computation. They are all parameterized by a class $\calC$ of Boolean functions. 
We show reductions between these problems that take time polynomial in the size of the functions {\em under the assumption} that the OR-substitutions can be computed in time polynomial in the sizes of the function and of the substitution.
In subsequent sections we show two well-known examples where this assumption is met: for deterministic and decomposable circuits (Section~\ref{sec:circuits}) and for query lineage (Section~\ref{sec:queries}).

Given a formula $F \in \calC$ over $n$ variables, the model counting problem asks for the number of models of $F$: 
\vspace*{.5em}

\begin{center}
\fbox{%
    \parbox{0.45\linewidth}{%
    \begin{tabular}{ll}
   Problem: & $\#\calC$ \\
   Description: &  \textit{Model Counting} \\
        Input: & $F \in \calC$\\
        Compute: & $\#F$
    \end{tabular}
    }}%
\end{center}

There is extensive literature on the model counting problem $\#\calC$~\cite{GomesSS:ModelCounting:2021}. We use two examples later in this paper. If $\calC$ is the class of positive, bipartite functions in disjunctive normal form, i.e., functions of the form $F = \bigvee_{(i,j)\in E} (X_i \wedge Y_j)$ where $E$ is a set of pairs $E \subseteq [n] \times [n]$, then $\#\calC$ is \#P-hard~\cite{ProvanB83}. If $\calC$ is the class of deterministic and decomposable Boolean circuits, then  $\#\calC$ is in FP~\cite{DarwicheM:KCMap:2002,DarwicheMSS:KnowledgeCompilation:2017}.

The fixed-size model counting problem asks for the number of models of $F$ of size $k$, 
for any $0\leq k \leq n$: 

\begin{center}
\fbox{%
    \parbox{0.45\linewidth}{%
    \begin{tabular}{ll}
    Problem: & $\#_{*}\calC$ \\  
    Description: & \textit{Fixed-Size Model Counting} \\       
        Input: & $F \in \calC$ over $n$ variables\\
        Compute: & $\#_{0,\ldots,n}F$
    \end{tabular}
    }}%
\end{center}

\vspace{60pt}

The Shapley value computation problem asks for the Shapley value of each variable in $F$: 

\begin{center}
    \fbox{%
    \parbox{0.45\linewidth}{%
    \begin{tabular}{ll}
    Problem: & $\shap(\calC)$ \\  
    Description: & \textit{Shapley Value Computation} \\      
        Input: & $F \in \calC$ \\
        Compute: & $\shap(F)$
    \end{tabular}
    }}%
\end{center}

Our main result gives polynomial-time reductions between the above three problems:

\begin{thm} 
\label{th:main}
  Given a class $\calC$ of Boolean functions, it holds:
\begin{itemize}
\item $\shap(\calC) \leq^P \#_*\widetilde{\calC}$
\vspace*{.25em}
\item $\#_*\calC \leq^P \#\widetilde{\calC}$
\vspace*{.25em}
\item $\#\calC  \leq^P \shap(\widetilde{\calC})$.
\end{itemize}
\end{thm}

In case $\calC$ OR-substitutes to itself, i.e., $\calC = \widetilde{\calC}$, the three problems $\#\calC$, $\#_*\calC$, and $\shap(\calC)$ become polynomial-time equivalent:
\begin{corollary}[Theorem~\ref{th:main}] 
\label{cor:main}
  Given a class $\calC$ of Boolean functions with $\calC = \widetilde{\calC}$, it holds:
$$\shap(\calC)\equiv^P \#_*\calC \equiv^P \#\calC$$
\end{corollary}

This result connects model counting to Shapley value computation. Whenever model counting is tractable for a class $\calC$ of Boolean functions that is closed under OR-substitutions, then the Shapley value computation is also tractable. We give here an immediate example; Sections~\ref{sec:circuits} and~\ref{sec:queries} provide two further examples. 

The class $\calC$ of {\em positive $\beta$-acyclic CNF} functions is trivially closed under OR-substitutions\footnote{The hypergraph of a CNF function has one node per variable and one hyperedge per clause. It is $\beta$-acyclic if there is no cycle in the hypergraph, nor in any sub-hypergraph. Substituting a variable by a disjunction of fresh variables preserves 
the structure of the CNF and of its hypergraph, except for replacing one node by several nodes that all occur in the same hyperedges as the replaced node.}. Furthermore, $\#\calC$ is in FP~\cite{Brault-BaronCM:CountSAT:2015}\footnote{Tractability holds even when removing the restriction on the functions being positive.}. Corollary~\ref{cor:main} then implies that $\shap(\calC)$ is also in FP.

There are two immediate generalizations of Theorem~\ref{th:main}. First, we may allow for polynomial-time transformations to accommodate the OR-substitutions. That is, the polynomial-time equivalence between the two problems holds whenever $\calC \approx^P \widetilde \calC$ holds and not only when $\calC = \widetilde \calC$ holds. Second, we may use substitutions beyond the OR-substitution considered here, such as AND-substitutions (more details are given at the end of Section~\ref{sec:main-results-functions}). \nop{An intriguing topic for future work is to understand under which substitutions our main result holds.}

\subsection{Proof of Theorem \ref{th:main}}
\label{sec:main_proof}

We separate the theorem into three lemmas:

\begin{lmm} \label{lemma:2}
  $\shap(\calC) \leq^P \#_*\widetilde{\calC}$
\end{lmm}

\begin{lmm} \label{lemma:1}
  $\#_*\calC \leq^P \#\widetilde{\calC}$
\end{lmm}

\begin{lmm} \label{lemma:3}
  $\#\calC \leq^P \shap(\widetilde{\calC})$.
\end{lmm}

\begin{proof}[Proof of Lemma~\ref{lemma:2}]
    Let $F\in \calC$ be a Boolean function. Our goal is to compute $\shap(F)$ in polynomial time, given an oracle for $\#_*\widetilde{\calC}$. 
 We use Eq.~\eqref{eq:shap:c} for the Shapley value and the following equality:
    $$\#_{k+1}F = \#_kF[X_i:=1] + \#_{k+1}F[X_i:=0]$$
  Then, Eq.~\eqref{eq:shap:c} becomes:
  \begin{align*}
    \shap(F,X_i) = & \sum_{k=0}^{n-1}c_k\left(\#_{k+1}F-\#_{k+1}F[X_i:=0]-\#_kF[X_i:=0]\right)
  \end{align*}
  Consider the function $\widetilde{F}$ that results from $F$ by replacing each variable by a fresh variable and the function $\widetilde{F}'$ that results from $F$ by replacing $X_i$ by the empty disjunction and each other variable by a fresh variable. 
  Clearly, $F$ admits OR-substitutions into 
  $\widetilde{F}$ and $\widetilde{F}'$, hence, $\widetilde{F},\widetilde{F}' \in \widetilde{\calC}$. The functions $\widetilde{F}$ and $\widetilde{F}'$ are isomorphic (i.e., identical up to renaming of the variables) to $F$ and respectively $F[X_i:=0]$, so model counting and fixed-size model counting is the same for $F$ and $\widetilde{F}$, and also for $F[X_i:=0]$ and $\widetilde{F}'$.
  We thus have access to an oracle to compute the quantities $\#_{k+1}F$, $\#_{k+1}F[X_i:=0]$, and $\#_kF[X_i:=0]$.  This means that we can compute $\shap(F,X_i)$ in polynomial time. 
\end{proof}


\begin{proof}[Proof of Lemma~\ref{lemma:1}] 
Let $F \in\calC$ be a Boolean function over the variables 
$\bm X = \set{X_1, \ldots, X_n}$.  Our goal
  is to compute $\#_{0,\ldots,n}(F)$ in polynomial time, given an oracle 
  for $\#\widetilde{\calC}$.  For a valuation
  $\theta : \bm X \rightarrow \set{0,1}$, we write $|\theta|$ for the
  number of variables $X_i$ s.t. $\theta(X_i) = 1$.  It follows:
  \begin{align*}
    \#_kF = & \sum_{\theta: |\theta|=k} F[\theta]
  \end{align*}
for $0 \leq k \leq n$.
  For each $\ell \in \mathbb{N}$, define:
  \begin{align*}
    F^{(\ell)} \defeq & F[X_1 := \bigvee_{j=1}^{\ell}Z_{1}^j, \ldots, X_n := \bigvee_{j=1}^{\ell}Z_{n}^j]
  \end{align*}
  where each $Z_{i}^j$ with  
  $i\in [n]$ and $j \in [\ell]$ is a fresh variable.
  It holds  $F^{(\ell)} \in \widetilde{\calC}$.
  Therefore,  we have access
  to an oracle for computing $\#F^{(\ell)}$.  We claim:
\begin{claim}
\label{cl1}
For each $\ell \in \mathbb{N}$, it holds:
  \begin{align}
\#F^{(\ell)} = & \sum_{k=0}^n (2^\ell-1)^k \#_kF \label{eq:claim:1}
  \end{align}
\end{claim}
  Claim~\ref{cl1} implies 
  Lemma~\ref{lemma:1} as follows. 
  We use Eq.~\eqref{eq:claim:1} for 
  $\ell \in [n+1]$ to form a
   system of $n+1$ linear equations with the 
  $n+1$ unknowns 
  $\#_0F, \ldots , \#_{n}F$. 
  The matrix of this system is a
  Vandermonde $(n+1)\text{-by-}(n+1)$ matrix, which is non-singular so we can compute its inverse~\cite{GolubVanLoan:Matrix:1996}.
  Hence, we can solve the linear system
  \begin{equation*}
        \underbrace{
      \begin{pmatrix}
      \#F^{(1)} \\ 
      \vdots \\
      \#F^{(n+1)} \\ 
      \end{pmatrix}
      }_{\text{known}}
    =    
        \underbrace{
      \begin{pmatrix}
      1& (2^1-1)^1 & \cdots & (2^1-1)^n \\ 
      & \vdots & \vdots & \vdots \\
      1& (2^{n+1}-1)^1 & \cdots & (2^{n+1}-1)^n \\ 
      \end{pmatrix}
      }_{\text{Vandermonde}}
      \underbrace{
      \begin{pmatrix}
      \#_0F \\ 
      \vdots \\
      \#_nF \\ 
      \end{pmatrix}
      }_{\text{unknown}}
  \end{equation*}
  and determine the values of $\#_0F, \ldots , \#_{n}F$ in polynomial time.  

  It remains to prove Claim~\ref{cl1}. 
  Let $\bm Z = \{Z_{i}^j | i\in [n]$
  and $j \in [\ell]\}$ be
  the set of variables of $F^{(\ell)}$.  By definition, it holds:
  \begin{align}
  \label{eq:subst_count}
    \#F^{(\ell)}=&\sum_{\varphi: \bm Z \rightarrow \set{0,1}}F^{(\ell)}[\varphi]
  \end{align}
  For each valuation $\varphi : \bm Z \rightarrow \set{0,1}$ we define
  the {\em induced valuation}
  $\theta_\varphi : \bm X \rightarrow \set{0,1}$ by setting
  $\theta_\varphi(X_i) = \varphi(\bigvee_{j=1}^{\ell}Z_{i}^j)$.  In other
  words, $\theta_\varphi(X_i)=1$ iff $\varphi$ evaluates
  $Z_{i}^1 \vee \cdots \vee Z_{i}^{\ell}$ to $1$.  Notice
  that:
  \begin{align}
    \forall \varphi:\bm Z \rightarrow \set{0,1},\ \ \  F^{(\ell)}[\varphi] = & F[\theta_\varphi] \label{eq:fresh_valuation1}\\
    \forall \theta: \bm X \rightarrow \set{0,1},\ \ \ |\setof{\varphi}{\theta_\varphi=\theta}|=&(2^\ell-1)^{|\theta|}
  \label{eq:fresh_valuation2}
  \end{align}
 We group each valuation $\varphi$ in 
  Eq.~\eqref{eq:subst_count} by its induced valuation $\theta_\varphi$:
    \begin{align*}
    \#F^{(\ell)}=&\sum_{\theta:\bm X \rightarrow \set{0,1}}\ \sum_{\varphi: \theta_\varphi = \theta}F^{(\ell)}[\varphi]\\
=& \sum_{\theta:\bm X \rightarrow \set{0,1}}\ \sum_{\varphi: \theta_\varphi = \theta}F[\theta]\TAB\TAB\TAB &\text{(by Eq.~\eqref{eq:fresh_valuation1})}\\
=& \sum_{\theta:\bm X \rightarrow \set{0,1}}(2^\ell-1)^{|\theta|}F[\theta]\TAB\TAB\TAB &\text{(by Eq.~\eqref{eq:fresh_valuation2})}\\
=& \sum_{k=0}^n \ \  \ \sum_{\theta:\bm X \rightarrow \set{0,1}: |\theta|=k}(2^\ell-1)^{k}F[\theta]\\
=& \sum_{k=0}^n(2^\ell-1)^{k}\#_kF
    \end{align*}
    This completes the proof of Claim~\ref{cl1}, which implies Lemma~\ref{lemma:1}.
\end{proof}

\begin{proof}[Proof of Lemma~\ref{lemma:3}] 
Let $F \in C$ be a Boolean function.
Our goal is to compute $\#F$ in polynomial time given an oracle to 
 $\shap(\widetilde{\calC})$.
Suppose $F$ has $n$ variables $\bm X = \set{X_1, \ldots, X_n}$.
  We fix $\ell \in \mathbb{N}$.
  For each variable $X_i$, let $F^{(\ell,i)}$ be the function
  obtained from $F$ by substituting $X_i$ with a fresh variable $Z_i$ 
  and every other variable $X_p$ with a disjunction of fresh variables
  $X_p := Z_{p}^{1} \vee \cdots \vee Z_{p}^{\ell}$.  
  The function $F$ admits OR-substitutions into $F^{(\ell,i)}$, hence,
  $F^{(\ell,i)} \in \widetilde{\calC}$.
  Using the  oracle for $\shap(\widetilde{\calC})$, we compute $\shap(F^{(\ell,i)},Z_i)$.  Then, using
  Eq.~\eqref{eq:shap:c} for the Shapley value and Eq.~\eqref{eq:claim:1}, we obtain:
  \begin{align*}
    \shap(F^{(\ell,i)},Z_i) = & \sum_{k = 0}^{n-1} c_k\left(\#_kF^{(\ell,i)}[Z_i:=1]-\#_kF^{(\ell,i)}[Z_i:=0]\right)\\
    = & \sum_{k=0}^{n-1}(2^\ell-1)^kc_k\left(\#_kF[X_i:=1]-\#_kF[X_i:=0]\right)
  \end{align*}
  Keeping $i$ fixed, we let $\ell$ iterate over $[n]$ to 
  form a system of $n$ equations with $n$ unknowns
  $\Gamma_kF \defeq c_k\left(\#_kF[X_i:=1]-\#_kF[X_i:=0]\right)$, 
  $k \in \{0, \ldots , n-1\}$. 

   \begin{equation*}
        \underbrace{
      \begin{pmatrix}
      \shap(F^{(1,i)},Z_i) \\ 
      \vdots \\
      \shap(F^{(n,i)},Z_i) \\ 
      \end{pmatrix}
      }_{\text{known}}
    =    
        \underbrace{
      \begin{pmatrix}
      (2^1-1)^0 & \cdots & (2^1-1)^{n-1} \\ 
      \vdots & \vdots & \vdots \\
      (2^n-1)^0 & \cdots & (2^n-1)^{n-1} \\ 
      \end{pmatrix}
      }_{\text{Vandermonde}}
      \underbrace{
      \begin{pmatrix}
      \Gamma_0F \\ 
      \vdots \\
      \Gamma_{n-1}F \\ 
      \end{pmatrix}
      }_{\text{unknown}}
  \end{equation*}

  The matrix of the equation system is a Vandermonde matrix, hence, nonsingular.
  We solve the system, and, since the constants $c_k$ are known and computable in polynomial time,
  we obtain all differences $\#_kF[X_i:=1]-\#_kF[X_i:=0]$.  We next show how to compute $\#F$ using these
  differences.  Let us keep $k$ fixed and sum these differences for $i \in [n]$.  We claim:
\begin{claim}
  \label{claim:2}
  For any $k \in \{0, \ldots , n-1\}$, it holds: 
  \begin{align*}
    \sum_{i=1}^{n}\left(\#_kF[X_i:=1]-\#_kF[X_i:=0]\right) =& (k+1)\#_{k+1}F - (n-k)\#_kF
  \end{align*}
  \end{claim}
  Claim~\ref{claim:2} follows from the following two equalities:
  \begin{align}
    \sum_{i=1}^n\#_kF[X_i:=1] =& (k+1)\#_{k+1}F \label{eq:count1}\\
    \sum_{i=1}^{n}\#_kF[X_i:=0] =& (n-k)\#_kF \label{eq:count2}
  \end{align}
    Equality~\eqref{eq:count1} holds as follows:
  \begin{align*}
    \sum_{i=1}^{n}\#_kF[X_i:=1] = & \sum_{i=1}^{n}\ \sum_{\theta:\bm X-\set{X_i}\rightarrow \set{0,1};|\theta|=k}F[\set{X_i:=1}\cup\theta]\\
= & \sum_{i=1}^{n}\ \sum_{\varphi:\bm X\rightarrow \set{0,1};|\varphi|=k+1;\varphi(X_i)=1}F[\varphi]\\
\overset{(*)}{=} & (k+1)\sum_{\psi:\bm X\rightarrow \set{0,1};|\psi|=k+1}F[\psi]\\
= & (k+1)\#_{k+1}F
  \end{align*}
 Equality $(*)$ holds because each valuation $\varphi$, which maps $X_i$ and $k$ other variables to $1$ and the remaining $n-k-1$ variables to $0$, is considered $k+1$ times when iterating over all $i\in[n]$. More precisely, let $T$ be the set of the indices of the $k+1$ variables set to $1$ in $\varphi$. Then, out of the $n$ iterations in the outer sum, the valuation $\varphi$ is only considered for $i\in T$.

  \nop{
  In words, for each term $\#_kF[X_i:=1]$ we need to sum over
  valuations $\theta$ that set $k$ variables other than $X_i$ to $1$.
  If we also include $X_i$, which is set to 1 in $F[X_i:=1]$, then we
  obtain a valuation $\varphi = \set{X_i:=1} \cup \theta$ that sets
  $k+1$ variables to $1$, including $X_i$.  Conversely, any valuation
  $\psi : \bm X \rightarrow \set{0,1}$ that sets $k+1$ variables to $1$
  can be obtained from $k+1$ different $\varphi$'s, since there are
  $k+1$ possible choices for the distinguished variable $X_i$.  
  }

 Equality~\eqref{eq:count2} above follows from a similar argument.
  \begin{align*}
    \sum_{i=1}^{n}\#_kF[X_i:=0] = & \sum_{i=1}^{n}\ \sum_{\theta:\bm X-\set{X_i}\rightarrow \set{0,1};|\theta|=k}F[\set{X_i:=0}\cup\theta]\\
= & \sum_{i=1}^{n}\ \sum_{\varphi:\bm X\rightarrow \set{0,1};|\varphi|=k;\varphi(X_i)=0}F[\varphi]\\
\overset{(**)}{=} & (n-k)\sum_{\psi:\bm X\rightarrow \set{0,1};|\psi|=k}F[\psi]\\
= & (n-k)\#_{k}F
  \end{align*}
 Equality $(**)$ holds because each valuation $\varphi$, which maps $X_i$ to $0$, $k$ other variables to $1$, and the remaining $n-k-1$ variables to $0$, is considered $n-k$ times when iterating over all $i\in[n]$. More precisely, let $T$ be the set of the indices of the $k$ variables set to $1$ in $\varphi$. Then out of the $n$ iterations in the outer sum, the valuation $\varphi$ is only considered for $i\in [n]\setminus T$, as for $i\in T$ the considered valuations have variable $X_i$ set to $0$.
  
\nop{
  The explanation of Equality $(**)$ is analogous to the one for Equality $(*)$. Let 
   $\varphi$ be a model of $F[\varphi]$ that maps exactly 
  the variables ${\bm Y} \subseteq {\bm X}$ with $|\bm Y| = k$
  to $1$.
  Within the scope of the double sum, 
  this valuation appears once for each choice for $X_i$ among the $n - k$ variables in ${\bm X}- {\bm Y}$. This means that the valuation appears exactly $n - k$ times.
}

  This completes the proof of Claim~\ref{claim:2}.  Thus, we
  have computed all $n$ differences $(k+1)\#_{k+1}F - (n-k)\#_kF$.
  The final step is the following.  Start by observing that
  $\#_0F = F[\bm 0]$, where $\bm 0$ is the valuation that sets all
  variables to 0.  Then, proceed inductively, computing $\#_kF$ for
  $k=\{1,\ldots, n\}$, using Claim~\ref{claim:2}, where we
  have already computed the left-hand side.
\end{proof}

\paragraph{AND-substitutions}

Theorem~\ref{th:main} also holds for AND-substitutions:
  \begin{align*}
    F^{(\ell)} \defeq & F[X_1 := \bigwedge_{j=1}^{\ell}Z_{1}^j, \ldots, X_n := \bigwedge_{j=1}^{\ell}Z_{n}^j]
  \end{align*}
  where each $Z_{i}^j$ with $i\in [n]$ and $j \in [\ell]$ is a fresh variable. To accommodate AND-substitutions, Claim~\ref{cl1} changes as follows: 
\begin{claim}
For each $\ell \in \mathbb{N}$, it holds:
  \begin{align*}
\#F^{(\ell)} = & \sum_{k=0}^n (2^\ell-1)^{n-k} \#_kF
  \end{align*}    
\end{claim}

\section{From Functions to Circuits}
\label{sec:circuits}

In general, Boolean functions do not admit polynomial-time satisfiability and model counting. Knowledge compilation is an approach that turns Boolean functions into equivalent representations that admit polynomial-time computation for a large number of tasks including model counting~\cite{DarwicheM:KCMap:2002,DarwicheMSS:KnowledgeCompilation:2017}. The price to pay is a possibly exponential time in the number of variables of the function to compute such an equivalent yet tractable representation. The tractability of well-known circuits, such as  OBDDs and d-DNNFs, relies on two key properties: determinism and decomposability.

We next recall the notion of a deterministic and decomposable circuit and then show that such circuits can efficiently accommodate OR-substitutions. This implies that the Shapley value can be computed in time polynomial in the size of such tractable circuits.

\subsection{Deterministic and Decomposable Circuits}

A {\em Boolean circuit} $G$ over a set $\bm X$ of variables is a directed acyclic graph where each node is one of the following {\em gates}:
\begin{itemize}
    \item A constant gate labeled with either 0 or 1;
    \item A variable gate labeled with a variable from $\bm X$; 
    \item A logic gate labeled with a Boolean connector $\wedge$ (and), $\vee$ (or), or $\neg$ (not).
\end{itemize}
The constant and variable gates have no incoming edges. The logic gates $\wedge$ and $\vee$ may have two or more incoming edges, and the logic gate $\neg$ has one incoming edge. There is one gate, called the output gate, that has no outgoing edge. The size of a circuit $G$, denoted by $|G|$, is the number of its gates (or the number of edges minus one). A valuation $\theta$ over $\bm X$ maps the circuit $G$ to $G[\theta]$, which is 0 or 1.

Boolean circuits are representations of Boolean functions. In this paper we are interested in Boolean circuits that satisfy the determinism and decomposability properties. Given a circuit $G$, a gate $g$ in $G$ defines the circuit $G_g$ that is $G$ where all gates that have no directed path to $g$ are removed. An $\vee$-gate $g$ is {\em deterministic} if  for every pair $(g_1,g_2)$ of distinct input gates of $g$, their circuits $G_{g_1}$ and  $G_{g_2}$ are disjoint: There is no valuation $\theta$ such that $G_{g_1}[\theta] = G_{g_2}[\theta] = 1$. An $\wedge$-gate $g$ is {\em decomposable} if for every pair $(g_1,g_2)$ of distinct input gates of $g$, their circuits $G_{g_1}$ and  $G_{g_2}$ have no variable in common. A circuit is deterministic if all its $\vee$-gates are deterministic and is decomposable if all its $\wedge$-gates are decomposable. 

\begin{example}
    Consider the circuit $(\neg X_1 \wedge X_2) \vee (X_1 \wedge X_3)$. It is deterministic as its only $\vee$-gate is deterministic: There is no valuation that maps both $\neg X_1 \wedge X_2$ and $X_1 \wedge X_3$ to 1, since the two functions are mutually exclusive. It is also decomposable since for both $\wedge$-gates have input gates whose circuits do not share variables. 
\end{example}

\subsection{Circuits under OR-substitutions}

Our main insight in this section is that the deterministic and decomposable circuits can efficiently accommodate OR-substitutions. Let $\calG$ be the class of deterministic and decomposable circuits and $\widetilde \calG$ be the class of circuits in $\calG$ where some variables are OR-substituted.

\begin{lemma}\label{lemma:dD-closure}
    $\widetilde \calG~\lesssim^P~\calG$.
\end{lemma}
More precisely, we can show the following for any deterministic and decomposable circuit $G$, a variable $X$ that occurs $k$ times in $G$, and distinct variables $Z_1,\ldots, Z_n$ that do not occur in $G$: A deterministic and decomposable circuit that represents $G$ under the OR-substitution $X\overset{\text{OR}}{\rightarrow} \bigvee_{i=1}^\ell Z_i$ can be computed in $O(|G|+k\ell)$
time.
This proves that the assumption made at the beginning of Section~\ref{sec:main-results-functions} holds for such circuits.

\begin{proof}
While the circuit $G_\vee(Z_1,\ldots,Z_\ell) = Z_1\vee \cdots \vee Z_\ell$ that replaces $X$ is not deterministic, it can be turned into an equivalent deterministic and decomposable circuit of size $O(\ell)$: 
    \begin{align*}
        G_\vee(Z_i,\ldots,Z_\ell) &= Z_i \vee (\neg Z_i \wedge (G(Z_{i+1},\ldots,Z_\ell))), \text{ for } i\in[\ell-1]\\
        G_\vee(Z_\ell) &= Z_\ell
    \end{align*}
Its negation $\neg G_\vee(Z_1,\ldots,Z_\ell)$ can be equivalently expressed as $\neg Z_1\wedge\cdots\wedge\neg Z_\ell$, which is both deterministic and decomposable, since $Z_1$ to $Z_\ell$ are distinct variables. Furthermore, substituting $X$ by $G_\vee$ and $\neg X$ by $\neg G_\vee$ does not violate the decomposability and determinism of the gates that are reached from $X$ and $\neg X$.
\end{proof}

The next theorem states that the Shapley value can be computed in polynomial time on deterministic and decomposable circuits. It is an immediate corollary of three results: (1) the well-known result on tractability of model counting for $\calG$~\cite{DarwicheM:KCMap:2002}; (2) Lemma~\ref{lemma:dD-closure} stating that OR-substitutions can be assimilated by any circuit in $\calG$ in FP; and (3) Theorem~\ref{th:main} conditioning the tractability of $\shap$ on the tractability of model counting for functions under OR-substitutions.

\begin{thm}\label{th:dD-circuits}
    $\shap(\calG)$ is in FP. \nop{That is, for any deterministic and decomposable circuit $G$, $\shap(G)$ can be computed in polynomial time in the size of $G$.}
\end{thm}

\section{From Functions to Queries}
\label{sec:queries}

We now lift our investigation of the Shapley value computation problem from (propositional) Boolean functions to (first-order) conjunctive queries. This is an application of our main result in Theorem~\ref{th:main}, enabled by the observation that the {\em lineage or provenance polynomial}~\cite{DBLP:conf/pods/GreenKT07} of a query is in fact a Boolean function.

One challenge in  our pursuit is to understand what is the counterpart of OR-substitutions at the query level. For this purpose, we introduce the  notion of {\em stretching} of a query and show that the lineage of the stretching of a CQ $Q$ is equivalent to the lineage of $Q$ under OR-substitutions. Furthermore, the two lineages can be transformed into one another in polynomial time. One caveat specific to this section is that the problems and the reductions used in the results below use data complexity\footnote{Under data complexity, the query is fixed and has constant size. The  complexity $O(|D|^{|Q|})$ is thus polynomial time, since the exponent $|Q|$ is the constant query size.}.

The main result of this section is the recovery of the dichotomy for Shapley value computation~\cite{LivshitsBKS21} using immediate derivations based on our main theorem and classical results for  model counting.

\subsection{Conjunctive Queries and Lineage}

We consider databases where some relations are {\em endogenous} while all others are {\em exogenous}. While we are interested in the contribution of the tuples from  endogenous relations to the answer of a query, we disregard the contribution of the tuples from exogenous relations.
Whenever we need to distinguish between the two kinds of relations, we annotate an endogenous relation $R$ as $R^n$ and an exogenous relation $R$ as $R^x$.

A Boolean Conjunctive Query (CQ) is:
\begin{align}
  Q = & \exists \bm x \bigwedge_{j\in[m]} R_j(\bm y_j) \label{eq:cq}
\end{align}
where $\bm x$ is the tuple of all variables in $Q$, $R_j(\bm y_j)$ are the atoms of $Q$ where $R_j$ is either an endogenous or an exogenous relation, and $\bm y_j \subseteq \bm x$ for $j\in[m]$. The size of $Q$, denoted by $|Q|$, is the number $m$ of its atoms. We denote by $at(x)$ the atoms with variable $x$, i.e., $at(x) =\{R_j(\bm y_j) | j\in[m], x\in \bm y_j\}$. To distinguish between variables in queries from those in Boolean functions, we write the former in lowercase and the latter in uppercase.

A CQ $Q$ is {\em hierarchical} if for any two query variables $x$ and $y$, one of the the following conditions hold: $at(x)\cap at(y)=\emptyset$, $at(x)\subseteq at(y)$, or $at(y)\subseteq at(x)$. A CQ $Q$ is {\em self-join-free} if there are no two atoms for the same relation.

For each database instance $\bm D$, the {\em lineage} 
$F_{Q, \bm D}$ of a CQ $Q$ over $\bm D$
 is a positive Boolean function in disjunctive normal form (DNF) over the variables $v(t)$ associated to the tuples $t$ in $\bm D$. Each clause in the lineage is a conjunction of $m$ variables, where $m$ is the number of relation atoms in $Q$. We define lineage recursively on the structure of a CQ ($\bm D$ is implicit and dropped from the subscript):
\begin{align*}
    F_{Q_1\vee Q_2} &= F_{Q_1}\vee F_{Q_2} &\hspace*{2em} F_{Q_1\wedge Q_2} &= F_{Q_1}\wedge F_{Q_2} \\
    F_{\exists x Q} &= \bigvee_{a\in\text{adom}(D)} F_{Q[a/x]}\\
    F_{R^n(t)} &=
    \begin{cases}
        v(t) & \text{if } t\in R\\
        0    & \text{otherwise}
    \end{cases} &\hspace*{2em}
    F_{R^x(t)} &=
    \begin{cases}
        1 & \text{if } t\in R\\
        0    & \text{otherwise}
    \end{cases}
\end{align*}
The lineage of a conjunction (disjunction) of two subqueries is the conjunction (disjunction) of their lineages. In case of an existential quantifier $\exists x$, we construct the disjunction of the lineages of all residual queries obtained by replacing the query variable $x$ by each value in the active domain (adom) of the database $\bm D$. Once all variables in an atom $R(t)$ are replaced by constants, we check whether the tuple $t$ of these constants is in the relation $R$. If it is not, then it does not contribute to the lineage (it is 0, or false). If it is, then 
we distinguish two cases. If $R$ is endogenous, then the Boolean variable $v(t)$ associated with the tuple $t$ is added to the lineage. If $R$ is exogenous, then we add instead 1 (or true) to signal that the variable $v(t)$ is not relevant for Shapley value computation.

The query $Q$ defines a class of Boolean functions consisting of the lineages of $Q$ over all databases $\bm D$:
\begin{align*}
  \calC_Q \defeq & \{ F_{Q,\bm D} \mid \bm D \mbox{ is a database instance} \}
\end{align*}

\subsection{Stretching Databases and Queries}

The following transformation is central to this section:

\begin{definition}\label{def:stretching}
Given an endogenous relation $R^n(y_1, \ldots, y_k)$ with attributes $y_1, \ldots, y_k$, its {\em stretching} is the relation $\widetilde R^n(y_0, y_1, \ldots, y_k)$. That is, we add one new attribute on the first position.

Given a CQ, where $\forall j\in[m]: \bm a_j\subseteq\bm a$ and $\forall j\in[p]: \bm b_j\subseteq\bm b$: 
\begin{align*}
    Q = \exists \bm a\ \exists \bm b \bigwedge_{j\in[m]} R^n_j(\bm a_j) \wedge \bigwedge_{j\in[p]} S^x_j(\bm b_j)
\end{align*}
its {\em stretching} is the CQ
\begin{align*}
  \widetilde Q = &  \exists \bm a\ \exists z_1\ldots\exists z_m \ \exists \bm b \bigwedge_{j\in[m]} R^n_j(z_j, \bm a_j) \wedge \bigwedge_{j\in[p]} S^x_j(\bm b_j)
\end{align*}
where $z_1, \ldots, z_m$ are fresh existential variables, one for every atom of an endogenous relation. 

%
\end{definition}

\begin{example}
The stretching of the non-hierarchical query
\begin{align}
    Q = \exists x \exists y \ R^n(x) \wedge S^x(x,y) \wedge T^n(y) \label{eq:nonh-query}
\end{align}
is
\begin{align}
  \widetilde Q = \exists x \exists y \exists z_1 \exists z_2\ R^n(z_1,x) \wedge S^x(x,y) \wedge T^n(z_2,y) \label{eq:nonh-stretching}
\end{align}
\end{example}

The stretching at the query level captures the OR-substitutions at the lineage level. That is, the lineage of $Q$ under OR-substitutions can be recovered via a polynomial-time transformation from the lineage of the stretching of $Q$ and vice versa. This shows that the assumption made at the beginning of Section~\ref{sec:main-results-functions} holds for lineage: We can construct in polynomial time\footnote{This is in polynomial time data complexity, so possibly exponential in the query size or equivalently in the arity of the clauses in the lineage.} a lineage for the stretched query from the lineage of the query under OR-substitutions.
The relationship between a CQ $Q$, its stretching $\widetilde Q$, their lineages $F_{Q,D}$ and $F_{\widetilde Q,\widetilde D}$ over databases 
$D$ and $\widetilde D$, and the function 
$\widetilde F_{Q,D}$ obtained from $F_{Q,D}$ by 
OR-substitution, is depicted below:

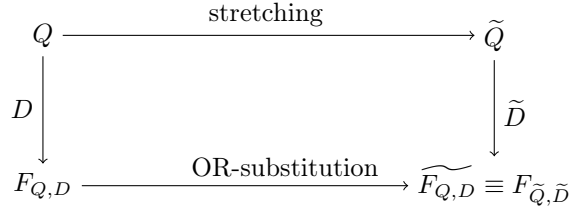
\begin{center}
\begin{tikzpicture}

\node (A) at (0, 0) {$Q$};
\node (B) at (6, 0) {$\widetilde Q$};
\node (C) at (0, -2) {$F_{Q,D}$};
\node (D) at (6, -2) {$\widetilde {F_{Q,D}} \equiv F_{\widetilde Q,\widetilde D}$};

\draw [->]  (A) edge (B) node [right = 2.7 of A, above] {stretching};
\draw [->]  (A) edge (C) node [below = 0.7 of A, left] {$D$};
\draw [->]  (B) edge (D) node [below = 0.7 of B, right] {$\widetilde D$};
\draw [->]  (C) edge (D) node [right = 2.7 of C, above] {OR-substitution};       
\end{tikzpicture}
\end{center}
In the bottom right node, the functions $\widetilde {F_{Q,D}}$ and $F_{\widetilde Q,\widetilde D}$ are equivalent and transformable into each other in polynomial time.
The above relationship implies a bidirectional polynomial-time transformation
between  $\widetilde{\calC_Q}$ and $\calC_{\widetilde Q}$:
\begin{lemma}\label{lm:lineage-stretching}
    $\widetilde{\calC_Q} \approx^P \calC_{\widetilde Q}$ holds for any CQ $Q$ and its stretching $\widetilde Q$.
\end{lemma}

\begin{example}
Consider the query 
$Q = \exists x R_1^n(x), R_2^n(x)$
and its stretching  $\widetilde{Q} = 
\exists x\exists z_1\exists z_2 R_1^n(z_1,x), R_2^n(z_2, x)$.
We depict below 
a database $\bm D$ consisting of 
the relations $R_1$ and $R_2$ and a database $\widetilde{\bm D}$
consisting of the stretched relations.
The variables $Y_i$ and $Z_i^j$ are associated to the database tuples. 
\begin{center}
\begin{tikzpicture}

\node (D) at (-0.8,0.7) {$\bm D:$};

\node (R) at (0,0.7) {$R_1^n(x)$};

\node (tableR) at (0,0) {
\begin{tabular}{c}
$x$ \\\hline
 $a_1$ \\
 $a_2$
\end{tabular}
};

\node (lineageR) at (-0.5,0) {
\begin{tabular}{c}
\phantom{a}\\
$Y_1:$  \\
$Y_2:$
\end{tabular}
};

\node (S) at (1.5,0.7) {$R_2^n(x)$};
\node (tableS) at (1.5,0) {
\begin{tabular}{c}
$x$ \\\hline
 $a_1$ \\
 $a_2$
\end{tabular}
};

\node (lineageS) at (1,0) {
\begin{tabular}{c}
\phantom{a}\\
$Y_3:$  \\
$Y_4:$
\end{tabular}
};


\begin{scope}[xshift= 13em, yshift= 5.5em]
\node (D') at (-1,-1) {$\widetilde{\bm D}:$};

\node (R) at (0.1,-1.05) {$R_1^n(z_1, x)$};

\node (S) at (2.2,-1.05) {$R_2^n(z_2,x)$};

\node (stretchedR) at (0.2,-2.7) {
\begin{tabular}{c@{\hskip 0.05in}c}
$z_1$ & $x$ \\\hline
$b_1^1$ & $a_1$ \\
 $\cdots$ & $\cdots$\\
 $b_1^{m}$& $a_{1}$ \\[5pt]
 $b_2^1$ & $a_2$ \\
 $\cdots$ & $\cdots$\\
$b_2^{n}$ & $a_{2}$ \\
\end{tabular}
};

\node (lineage_stretched_R) at (-0.8,-2.7) {
\begin{tabular}{c}
\phantom{a}\\
$Z_1^1:$  \\
$\cdots$ \\
$Z_1^{m}:$ \\[5pt]
$Z_2^1:$ \\
$\cdots$ \\
$Z_2^{n}:$
\end{tabular}
};

\begin{scope}[xshift=30pt]
\node (stretchedS) at (1.5,-2.7) {
\begin{tabular}{c@{\hskip 0.05in}c}
$z_2$ & $x$ \\\hline
$c_1^1$ & $a_1$ \\
 $\cdots$ & $\cdots$\\
 $c_1^{p}$& $a_1$ \\[5pt]
 $c_2^1$ & $a_2$ \\
 $\cdots$ & $\cdots$\\
$c_2^{q}$ & $a_2$ \\
\end{tabular}
};

\node (lineage_stretched_S) at (0.65,-2.7) {
\begin{tabular}{c}
\phantom{a}\\
$Z_3^1:$  \\
$\cdots$ \\
$Z_3^{p}:$ \\[5pt]
$Z_4^1:$ \\
$\cdots$ \\
$Z_4^{q}:$
\end{tabular}
};
\end{scope}
\end{scope}
\end{tikzpicture}
\end{center}

The lineage of
$Q$ over $\bm D$
is $F_{Q,\bm D} = 
(Y_1 \wedge Y_3) 
\vee 
(Y_2 \wedge Y_4)$, hence 
$F_{Q,\bm D} \in \calC_{Q}$.
The lineage of
$\widetilde{Q}$ over $\widetilde{\bm D}$
is 
$F_{\widetilde{Q},\widetilde{\bm D}}
=\bigvee_{i\in [m], j\in [p]}
(Z_1^i \wedge Z_3^j)
\vee
\bigvee_{i\in [n], j\in [q]}
(Z_2^i \wedge Z_4^j)$, 
hence $F_{\widetilde{Q},\widetilde{\bm D}} \in \calC_{\widetilde{Q}}$. 
Under the OR-substitution $\theta = 
\{
Y_1:= \bigvee_{i = 1}^m Z_1^i,
Y_2:= \bigvee_{i = 1}^n Z_2^i,$
$Y_3:= \bigvee_{i = 1}^p Z_3^i,
Y_4:= \bigvee_{i = 1}^q Z_4^i
\}$, we get 
$F_{Q,\bm D}[\theta] = 
((\bigvee_{i = 1}^m Z_1^i) \wedge 
(\bigvee_{i = 1}^p Z_3^i)) 
\vee 
((\bigvee_{i = 1}^n Z_2^i) \wedge 
(\bigvee_{i = 1}^q Z_4^i))$. 
It holds $F_{Q,\bm D}[\theta]\in 
\widetilde{\calC_Q}$.
Observe that  
$F_{Q,\bm D}[\theta] \equiv F_{\widetilde{Q},\widetilde{\bm D}}$ and 
can be transformed into one another in quadratic time 
using the distributivity law for $\wedge$ over $\vee$ (the time is exponential in the number of endogenous relations).
\end{example}

Lemma~\ref{lm:lineage-stretching} immediately implies the following polynomial-time equivalences between the three problems introduced in Section~\ref{sec:main-results-functions}, now over classes of query lineage:

\begin{corollary}[of Lemma~\ref{lm:lineage-stretching}]
\label{cor:lineage-stretching_corollary}
For any CQ $Q$ and its stretching $\widetilde{Q}$, the following polynomial-time equivalences hold:
\begin{itemize}
\item $\shap(\widetilde{\calC_Q}) \equiv^P \shap(\calC_{\widetilde{Q}})$ 
\item $\#\widetilde{\calC_Q}  \equiv^P \#\calC_{\widetilde{Q}}$
\item $\#_*\widetilde{\calC_Q}  \equiv^P \#_*\calC_{\widetilde{Q}}$
\end{itemize}
\end{corollary}
For instance, if we want to compute $\shap(\widetilde F)$ for $\widetilde F \in \widetilde{\calC_Q}$, i.e., for $Q$'s lineage under OR-substitutions, and have an oracle for $\shap(\calC_{\widetilde Q})$, i.e., for computing the Shapley values for the lineage of $Q$'s stretching ${\widetilde Q}$, we can first transform  $\widetilde F$ in polynomial time into an equivalent function $F \in \calC_{\widetilde Q}$ and then compute $\shap(F)$ using the oracle. Since $\widetilde F \equiv F$, we have $\shap(\widetilde F) = \shap(F)$.

Query stretching preserves the hierarchical property:
\begin{lemma}\label{lm:hierarchical-stretching}
    A CQ $Q$ is hierarchical iff its stretching $\widetilde Q$ is hierarchical.
\end{lemma}

\subsection{Dichotomy for Self-Join-Free CQs}
We prove the following dichotomy using our polynomial-time equivalences and lineage transformations:

\begin{thm}[\cite{LivshitsBKS21}]
\label{th:dichotomy-hierarchical-query}
    Let $Q$ be a self-join-free CQ. If $Q$ is hierarchical, then $\shap(\calC_Q)$ is in FP, otherwise it is FP$^{\text{\#P}}$-hard.
\end{thm}

The hardness result holds for specific classes of databases, where we can choose conveniently the endogenous and exogenous relations, whereas the tractability result holds for any database.
We first focus on hardness and later on tractability.

\paragraph{Hardness} We show that for any non-hierarchical CQ $Q$, there are specific classes of databases for which $\shap(\calC_Q)$ is FP$^{\text{\#P}}$-hard. We first show the hardness for the smallest non-hierarchical CQ  and then generalize to arbitrary non-hierarchical CQs.

Let us consider the smallest non-hierarchical CQ in Eq.~\eqref{eq:nonh-query} and its stretching in Eq.~\eqref{eq:nonh-stretching}, where we choose conveniently the relations $R$ and $T$ to be endogenous, while the relation $S$ be exogenous.

The class $\calC_Q$ consists of all positive bipartite functions in disjunctive normal form: $\bigvee_{(i,j)\in S} X_i\wedge Y_j$, where $X_i$ annotates tuple $R(i)$ and $Y_j$ annotates tuple $T(j)$. Any such function can be obtained by appropriately picking $R$ and $T$ for the sets of variables $X_i$ and $Y_j$, and $S$ to encode its clauses. We next use a prior result on the $\#P$-hardness for model counting for this class of functions~\cite{ProvanB83}:

\begin{align*}
    &\  \#\calC_Q  \leq^P \shap(\widetilde{\calC_Q}) & \text{(by Theorem~\ref{th:main})} \\
    \Rightarrow&\  \#\calC_Q  \leq^P \shap(\calC_{\widetilde Q})  & \text{(by Corollary~\ref{cor:lineage-stretching_corollary})}\\
    \Rightarrow&\  \#\calC_Q  \leq^P \shap(\calC_Q)  & \text{(by Claim~\ref{claim:equivalence-lineage-hierarchical} below)}\\
    \Rightarrow&\  \shap(\calC_Q) \text{ is FP}^{\text{\#P}}\text{-hard} & 
    \text{ (}\#\calC_Q \text{ is } \#P\text{-hard}~\text{\cite{ProvanB83})}
\end{align*}

\nop{
When applied to the lineage class $\calC_Q$, Theorem~\ref{th:main} states that $\#\calC_Q  \leq^P \shap(\widetilde{\calC_Q})$. From Lemma~\ref{lm:lineage-stretching}, we have $\widetilde{\calC_Q} = \calC_{\widetilde Q}$. This implies that $\#\calC_Q  \leq^P \shap(\calC_{\widetilde Q})$. Since $\#\calC_Q$ is $\#P$-hard for the class of all PP2DNF functions~\cite{ProvanB83}, this means that $\shap(\calC_{\widetilde Q})$ is FP$^{\#P}$-hard. In Claim~\ref{claim:equivalence-lineage-hierarchical} below, we show that $\calC_{\widetilde Q} = \calC_{Q}$, which concludes the hardness of the query in Eq.~\eqref{eq:nonh-query}.
}

\begin{claim}\label{claim:equivalence-lineage-hierarchical}
    $\calC_{\widetilde Q} = \calC_{Q}$ holds for the non-hierarchical query $Q$ in Eq.~\eqref{eq:nonh-query} and its stretching $\widetilde Q$ in Eq.~\eqref{eq:nonh-stretching}.
\end{claim}

The proof of Claim~\ref{claim:equivalence-lineage-hierarchical} is in  Appendix~\ref{sec:claim:equivalence-lineage-hierarchical}.

The generalization to arbitrary non-hierarchical CQs is as in prior work~\cite{DalviS04,LivshitsBKS21}. We reduce the computation of $Q$ in Eq.\eqref{eq:nonh-query} over any database $\bm D$ to the computation of any non-hierarchi\-cal query $Q'$ over a specifically-designed database $\bm D'$ 
constructed from $\bm D$.

By definition, the non-hierarchical query $Q$ has two variables $x$ and $y$ such that $at(x)\cap at(y)\neq\emptyset$, $at(x)\not\subseteq at(y)$, and $at(y)\not\subseteq at(x)$. We construct $\bm D'$ as follows. We pick two distinct atoms in $Q$, call them $R$ and $T$, such that: $R$ has $x$ and not $y$, and $T$ has $y$ and not $x$. We make the relations of these two atoms endogenous and all other relations exogenous. The values for all other variables are set to the same constant, say $1$, while the values of $x$ in $R$ and of $y$ in $T$ are precisely those in the database $\bm D$. The $x$ ($y$) columns in the other relations in $\bm D'$ are copies of the corresponding columns in $R$ ($T$), so the semi-joins of $R$ ($T$) with its copies do not alter $R$ ($T$). Then, the lineage of $Q$ and $\widetilde Q$ over $\bm D$ and respectively $\bm D'$ is the same: $F_{Q,D} = F_{Q',D'}$. The hardness of $\#\calC_Q$ thus transfers to $\#\calC_{Q'}$.

\paragraph{Tractability}
We show that $\shap(\calC_Q)$ is in FP for any hierarchical CQ $Q$. 
We use that $\#\calC_Q$ is tractable for any hierarchical  $Q$~\cite{OlteanuH:PDB:2008}:
\begin{claim}
\label{claim:count_star_hierarchical_FP}
For any hierarchical CQ $Q$,  $\#_*\calC_Q$ is in FP. 
\end{claim}

\begin{proof}
\begin{align*}
& \ Q \text{ is hierarchical} \\
\Rightarrow &\  \widetilde{Q} \text{ is hierarchical} & \text{(by Lemma~\ref{lm:hierarchical-stretching})}\\
\Rightarrow &\  \#\calC_{\widetilde{Q}} \text{ is in FP} & \text{(by \cite{OlteanuH:PDB:2008})}\\
\Rightarrow &\  \#\widetilde{\calC_{Q}} \text{ is in FP} & \text{(by Corollary~\ref{cor:lineage-stretching_corollary})}\\
\Rightarrow &\  \#_*\calC_Q \text{ is in FP} & \text{(by Theorem~\ref{th:main})}
\end{align*}
\end{proof}

Tractability of $\shap(\calC_Q)$ is now an immediate implication: 
\begin{align*}
&\  Q \text{ is hierarchical} \\
\Rightarrow &\  \widetilde{Q} \text{ is hierarchical} & \text{(by Lemma~\ref{lm:hierarchical-stretching})}\\
\Rightarrow &\  \#_*\calC_{\widetilde{Q}} \text{ is in FP} & \text{(by Claim~\ref{claim:count_star_hierarchical_FP})}\\
\Rightarrow &\  \#_*\widetilde{\calC_{Q}} \text{ is in FP} & \text{(by Corollary~\ref{cor:lineage-stretching_corollary})}\\
\Rightarrow &\  \shap(\calC_Q) \text{ is in FP} & \text{(by Theorem~\ref{th:main})}
\end{align*}

\nop{
Let $\widetilde Q$ be the stretching of $Q$. Following Lemma~\ref{lm:hierarchical-stretching}, $\widetilde Q$ is also hierarchical.
When applied to the lineage class $\calC_Q$, Theorem~\ref{th:main} states that $\#_*\calC_Q \leq^P \#\widetilde{\calC_Q}$. From Lemma~\ref{lm:lineage-stretching}, we have $\widetilde{\calC_Q} = \calC_{\widetilde Q}$. This implies that $\#_*\calC_Q \leq^P \#\calC_{\widetilde Q}$. Since $\#\calC_Q$ is in FP for any hierarchical query $Q$~\cite{OlteanuH:PDB:2008}, and therefore $\#\calC_{\widetilde Q}$ is also in FP, this means that $\#_*\calC_Q$ is in FP. In Claim~\ref{claim:size-count} below, we show that $\#_*\calC_{\widetilde Q}$ is also in FP. By Theorem~\ref{th:main}, $\shap(\calC_Q) \leq^P \#_*\widetilde{\calC_Q}$, which becomes $\shap(\calC_Q) \leq^P \#_*\calC_{\widetilde Q}$ by Lemma~\ref{lm:lineage-stretching}. Therefore, $\shap(\calC_Q)$ is in FP.

\begin{claim}\label{claim:size-count}
    Let a hierarchical query $Q$ and its stretching $\widetilde Q$. Then, 
    $\#_*\calC_{\widetilde Q}$ is in FP.
\end{claim}

Proof of the claim comes here.

[We know: $\#\calC_{\widetilde Q}$ is in FP;  $\#_*\calC_Q$ is in FP; $\#\calC_Q$ is in FP]
}

\paragraph{Discussion} The above hardness proof is significantly simpler than the original one~\cite{LivshitsBKS21}, which solves several instances of computing the number of independent sets of a given bipartite graph and assembles them in a full-rank set of linear equations. In fact, the original proof questions\footnote{First paragraph in the proof of Proposition 4.6~\cite{LivshitsBKS21}: "It is not at all clear to us how such an approach can work in our case and, indeed, our proof is more involved".} whether a simple proof based on the hardness of model counting for positive bipartite DNF, as used to show the hardness of the non-hierarchical queries over probabilistic databases and also used in our proof above, is even possible. Our result settles this question in the affirmative.

\section{Conclusion and Future Work}

In this paper we give a polynomial-time equivalence between computing Shapley values  and model counting for any class of Boolean functions that are closed under substitutions of variables with disjunctions of fresh variables. This result settles an open problem raised in prior work. We also show two direct applications of our result: tractability of Shapley value computation for deterministic and decomposable circuits and the dichotomy for Shapley value computation in case of self-join-free Boolean conjunctive queries.

We conjecture that our work can be instrumental to show that the dichotomy for unions of conjunctive queries in probabilistic databases~\cite{DalviS12} also applies to Shapley value computation. Furthermore, we would like to understand the impact of more complex substitutions on the tractability of both model counting and of Shapley value computation.

\paragraph{Acknowledgements.}
Ahmet Kara and Dan Olteanu would like to acknowledge Daniel Deutch, who introduced them to the topic of Shapley values in databases.

\bibliographystyle{abbrv}
\bibliography{bibliography}

\appendix
\section{Missing Details in Section~\ref{sec:prelims}}
\label{app:prelims}
\subsection{Proof of Proposition~\ref{prop:alternative_shap}}
\textsc{Proposition} \ref{prop:alternative_shap} (\cite{LivshitsBKS21} page 11, adapted).
\textit{The {\em Shapley value} of a variable $X_i$ of a Boolean function $F$ is:
\begin{align*}
  \shap(F,X_i) = & \sum_{k=0}^{n-1} c_k \left(\#_kF[X_i:=1]-\#_kF[X_i:=0]\right)
\end{align*}
%
where $c_k = \frac{k! (n-k-1)!}{n!}$.
}

\medskip
We show this proposition as follows:
\begin{align*}
  \shap(F,X_i) &\overset{(a)}{=}  \frac{1}{n!} 
  \sum_{\Pi \in S_n} \left(F[\Pi^{<i} \cup \set{i}]-F[\Pi^{<i}]\right) \\  
    & \overset{(b)}{=}  \frac{1}{n!} \sum_{T \subseteq [n]-\{i\}} |T|! (n-|T|-1)!  
    \left(F[T \cup\{i\}] - F[T]\right) \\ 
        & \overset{(c)}{=}  \frac{1}{n!} \sum_{k = 0}^{n-1} |k|! (n-|k|-1)!  
    \left( \#_kF[X_i:=1]-\#_kF[X_i:=0] \right) \\
    & \overset{(d)}{=}  \sum_{k=0}^{n-1} c_k \left(\#_kF[X_i:=1]-\#_kF[X_i:=0]\right)
\end{align*}
Equality (a) holds by definition.
We obtain Equality (b) by grouping the sum by possible sets $T \subseteq [n]-\{i\}$ and scaling the result of $F[T \cup\{i\}] - F[T]$
by the number of permutations
of the set
$\{1, \ldots , n\}$
that start with the values in $T$ followed by $i$. Observe that $|T|! (n-|T|-1)!$ is the number of permutations of the set 
$\{1, \ldots , n\}$ that start with 
the values in $T$ followed by $i$. 
To obtain Equality (c), we iterate over the sizes of possible 
sets $T \subseteq [n]-\{i\}$
and observe that the number of sets 
$T \subseteq [n]-\{i\}$ of size $k$ such that 
$F[T \cup \{i\}] = 1$ 
is exactly $\#_kF[X_i:=1]$; similarly, 
the number of sets $T$ of size $k$ such that $F[T] = 1$ is 
$\#_kF[X_i:=0]$.
We obtain Equality (d) by moving $\frac{1}{n!}$ inside the sum and
replacing $\frac{k! (n-k-1)!}{n!}$ by $c_k$.

\subsection{Proof of Proposition~\ref{prop:sum_shap}}
\textsc{Proposition}~\ref{prop:sum_shap}.
\textit{For any Boolean function $F$, it holds
\begin{align*}
  \sum_{i\in[n]}\shap(F,X_i) = & F[\bm 1] - F[\bm 0]
\end{align*}
where $\bm 1$ is the valuation that maps all variables to $1$, and
$\bm 0$ the valuation that maps all variables to $0$.
}

\medskip
We show this proposition as follows:
\begin{align*}
  \sum_{i = 1}^n \shap(F,X_i) \overset{(a)}{=} & \sum_{i=1}^n \sum_{k=0}^{n-1} c_k \left(\#_kF[X_i:=1]-\#_kF[X_i:=0]\right) \\
  = & \sum_{k=0}^{n-1} \sum_{i= 1}^n c_k  \left( \#_kF[X_i:=1]-\#_kF[X_i:=0]\right) \\
  \overset{(b)}{=} & \sum_{k=0}^{n-1} \left(c_k
  (k+1) \#_{k+1}F -c_k (n-k)\#_kF\right) \\  
= &\  c_0\cdot \#_1 F - \underline{c_0 \cdot n\cdot \#_0 F}\ + \\
& \ c_1 \cdot 2 \cdot \#_2 F - c_1\cdot (n-1)\cdot \#_1 F \ + \cdots + \\
& \ \overline{c_{n-1}\cdot n\cdot \#_n F} - c_{n-1}\cdot  \#_{n-1} F \\
\overset{(c)}{=} &\  \overline{c_{n-1}\cdot n\cdot \#_n F} -  
\underline{c_0 \cdot n\cdot \#_0 F} + \\
& \ \sum_{k=0}^{n-2}\left( c_{k}(k+1) \#_{k+1}F - c_{k+1} (n - k -1) \#_{k+1} F \right)\\
%
\overset{(d)}{=} &\ 
c_{n-1}\cdot n\cdot \#_n F -  
c_0 \cdot n\cdot \#_0 F\\
\overset{(e)}{=} &\ F[\bm 1] - F[\bm 0] 
\end{align*}
Equality (a) uses the Shapley value characterization given in 
Proposition~\ref{prop:alternative_shap}.
Equality (b) follows from the two 
Equalities 
\eqref{eq:count1} and \eqref{eq:count2}
in Section \ref{sec:main_proof}.
We obtain Equality (c) by regrouping the terms on the left-hand side: We keep  
$c_{n-1}\cdot n\cdot \#_n F -  
c_0 \cdot n\cdot \#_0 F$
outside the scope of 
the sum and pair  
the terms $c_{k}(k+1) \#_{k+1}F$ and $c_{k+1} (n - k -1) \#_{k+1} F$ for $0 \leq k \leq n-2$
within the scope of the sum.
Equality (d) holds, since for each $k$, 
the two terms within the scope of the sum cancel each other.
This cancelling is due to the following equalities: $c_{k}(k+1) = 
\frac{k!(n-k-1)!}{n!} (k+1)= 
\frac{(k+1)!(n-k-1)!}{n!}= 
\frac{(k+1)!(n-k-2)!}{n!} (n - k -1)
= c_{k+1} (n - k -1)$.
Equality (e) follows from the equalities
$c_{n-1}\cdot n = c_{0}\cdot n = \frac{(n-1)!}{n!}n = 1$
and the observation that  $F$ can have at most one model of size $n$ and at most one model of size $0$.

\section{Missing Details in Section~\ref{sec:queries}}
\label{app:queries}
We introduce notation used in the following.
Given a relation $R$ over some attributes
$(y_1, \ldots, y_n)$, we write $(y_1: a_1, \ldots, y_n: a_n)$
to denote a tuple in $R$ where the $y_i$ value is $a_i$ for $i \in [n]$.

\subsection{Proof of Claim~\ref{claim:equivalence-lineage-hierarchical}}
\label{sec:claim:equivalence-lineage-hierarchical}
\textsc{Claim} \ref{claim:equivalence-lineage-hierarchical}.
 \textit{$\calC_{\widetilde Q} = \calC_{Q}$ holds for the non-hierarchical query $Q$ in Eq.~\eqref{eq:nonh-query} and its stretching $\widetilde Q$ in Eq.~\eqref{eq:nonh-stretching}.}

\medskip

We first illustrate how we can construct 
databases to show that each 
lineage 
in $\calC_{Q}$ 
is also a lineage in 
$\calC_{\widetilde Q}$
and vice versa.

\begin{example}
Consider the following database $\bm D$, 
where the variables $Y_i$ preceding the tuples in endogenous relations are associated to the tuples.

\begin{center}
\begin{tikzpicture}

\node (D) at (-2,0.7) {$\bm D:$};

\node (R) at (-1,0.7) {$R^n(x)$};

\node (tableR) at (-1,0) {
\begin{tabular}{c}
$x$ \\\hline
 $a_1$ \\
 $a_2$
\end{tabular}
};

\node (lineageR) at (-1.5,0) {
\begin{tabular}{c}
\phantom{a}\\
$Y_1:$  \\
$Y_2:$
\end{tabular}
};

\node (S) at (1.5,0.7) {$S^x(x,y)$};
\node (tableS) at (1.5,0) {
\begin{tabular}{cc}
$x$ & $y$ \\\hline
 $a_1$ & $b_1$ \\
 $a_2$ & $b_2$
\end{tabular}
};

\node (T) at (4,0.7) {$T^n(y)$};
\node (tableT) at (4,0) {
\begin{tabular}{c}
$y$ \\\hline
 $b_1$ \\
 $b_2$
\end{tabular}
};

\node (lineageT) at (3.5,0) {
\begin{tabular}{c}
\phantom{a}\\
$Y_3:$  \\
$Y_4:$
\end{tabular}
};

\end{tikzpicture}
\end{center}
The lineage of $Q$ over $\bm D$ is
$F_{Q,\bm D} = 
(Y_1 \wedge Y_3) \vee
(Y_2 \wedge Y_4)$.
It holds $F_{Q,\bm D} \in \calC_Q$. 
Now, we construct from $\bm D$ a database 
$\widetilde{\bm D}$ such that 
$F_{Q,\bm D}$ is the lineage 
of $\widetilde{Q}$ over
$\widetilde{\bm D}$.
The idea is to assign to the fresh attributes added due to stretching 
a dummy value $d$: 

\begin{center}
\begin{tikzpicture}
\node (D) at (-2, -1.2) {$\widetilde{\bm D}:$};

\node (R) at (-1,-1.2) {$\widetilde{R}^n(z_1,x)$};

\node (tableR) at (-1,-2) {
\begin{tabular}{cc}
$z_1$ & $x$ \\\hline
$d$ & $a_1$ \\
$d$ & $a_1$
\end{tabular}
};

\node (lineageR) at (-1.7,-2) {
\begin{tabular}{c}
\phantom{a}\\
$Y_1:$  \\
$Y_2:$
\end{tabular}
};

\node (S) at (1.5,-1.2) {$S^x(x,y)$};
\node (tableS) at (1.5,-2) {
\begin{tabular}{cc}
$x$ & $y$ \\\hline
 $a_1$ & $b_1$ \\
 $a_2$ & $b_2$
\end{tabular}
};

\node (T) at (4,-1.2) {$\widetilde{T}^n(z_2,y)$};
\node (tableT) at (4,-2) {
\begin{tabular}{cc}
$z_2$ & $y$ \\\hline
 $d$ & $b_1$ \\
 $d$ & $b_2$
\end{tabular}
};

\node (lineageT) at (3.3,-2) {
\begin{tabular}{c}
\phantom{a}\\
$Y_3:$  \\
$Y_4:$
\end{tabular}
};
\end{tikzpicture}
\end{center}

Now, consider the following database
$\widetilde{\bm D}'$ with stretched relations:  
\begin{center}
\begin{tikzpicture}
\node (D) at (-2, -1.2) 
{$\widetilde{\bm D}'$:};

\node (R) at (-1,-1.2) {$\widetilde{R}^n(z_1,x)$};

\node (tableR) at (-1,-2) {
\begin{tabular}{cc}
$z_1$ & $x$ \\\hline
$d_1$ & $a$ \\
$d_2$ & $a$
\end{tabular}
};

\node (lineageR) at (-1.8,-2) {
\begin{tabular}{c}
\phantom{a}\\
$Y_1:$  \\
$Y_2:$
\end{tabular}
};

\node (S) at (1.5,-1.2) {$S^x(x,y)$};
\node (tableS) at (1.5,-1.8) {
\begin{tabular}{cc}
$x$ & $y$ \\\hline
 $a$ & $b$
\end{tabular}
};

\node (T) at (4,-1.2) {$\widetilde{T}^n(z_2,y)$};
\node (tableT) at (4,-2) {
\begin{tabular}{cc}
$z_2$ & $y$ \\\hline
 $d_1$ & $b$ \\
 $d_2$ & $b$
\end{tabular}
};

\node (lineageT) at (3.2,-2) {
\begin{tabular}{c}
\phantom{a}\\
$Y_3:$  \\
$Y_4:$
\end{tabular}
};
\end{tikzpicture}
\end{center}
The lineage of $\widetilde{Q}$ over 
$\widetilde{\bm D}'$ is
$F_{\widetilde{Q},\widetilde{\bm D}'} = 
(Y_1 \wedge Y_3) \vee
(Y_1 \wedge Y_4) \vee
(Y_2 \wedge Y_3) \vee
(Y_2 \wedge Y_4)$.
It holds $F_{\widetilde{Q},
\widetilde{\bm D}'} \in \calC_{\widetilde{Q}}$. 
We construct now from 
$\widetilde{\bm D}'$ a database 
${\bm D}'$ such that
$F_{\widetilde{Q},
\widetilde{\bm D}'}$ is 
is a lineage 
of $Q$ over
${\bm D}'$.
The idea is to represent tuples 
over $(z_1,x)$ and $(z_2,y)$ as
single (composite) values over $x$
and respectively $y$ and construct $S$ such that the combinations of $(z_1,x)$ and $(z_2,y)$ remain the same as in $\widetilde{\bm D}'$: 

\begin{center}
\begin{tikzpicture}
\node (D) at (-2, -1.2) 
{${\bm D}'$:};

\node (R) at (-1,-1.2) {$R^n(x)$};

\node (tableR) at (-1,-2) {
\begin{tabular}{c}
$x$  \\\hline
$(d_1,a)$  \\
$(d_2,a)$ 
\end{tabular}
};

\node (lineageR) at (-1.8,-2) {
\begin{tabular}{c}
\phantom{a}\\
$Y_1:$  \\
$Y_2:$
\end{tabular}
};

\node (S) at (2,-1.2) {$S^x(x,y)$};
\node (tableS) at (2,-2.4) {
\begin{tabular}{cc}
$x$ & $y$ \\\hline
 $(d_1, a)$ & $(d_1,b)$\\
  $(d_1, a)$ & $(d_2,b)$ \\
   $(d_2, a)$ & $(d_1,b)$ \\
  $(d_2, a)$ & $(d_2,b)$
\end{tabular}
};

\node (T) at (5.2,-1.2) {$T^n(y)$};
\node (tableT) at (5.2,-2) {
\begin{tabular}{c}
 $y$ \\\hline
 $(d_1, b)$ \\
 $(d_2, b)$
\end{tabular}
};

\node (lineageT) at (4.4,-2) {
\begin{tabular}{c}
\phantom{a}\\
$Y_3:$  \\
$Y_4:$
\end{tabular}
};
\end{tikzpicture}
\end{center}
\end{example}

Next, we prove Claim~\ref{claim:equivalence-lineage-hierarchical} formally. 
Consider the non-hierarchical CQ 
$Q = \exists x \exists y \ R^n(x) \wedge S^x(x,y) \wedge T^n(y)$
in 
Eq.~\eqref{eq:nonh-query} and its stretching 
$\widetilde{Q} = \exists x \exists y \exists z_1 \exists z_2\ R^n(z_1,x) \wedge S^x(x,y) \wedge T^n(z_2,y)$
in Eq.~\eqref{eq:nonh-stretching}.
We first show that $C_Q \subseteq C_{\widetilde{Q}}$ and 
then we show $C_{\widetilde{Q}} \subseteq C_{Q}$. 

\subsubsection{$C_Q \subseteq C_{\widetilde{Q}}$}
\label{sec:RST_lineage_in_substituted_RST_lineage}
Consider the lineage $F_{Q,\bm D} \in \calC_Q$
for a database $\bm D = \{R^n,S^x,T^n\}$. We show that 
$F_{Q,\bm D} \in \calC_{\widetilde{Q}}$.
We construct from $\bm D$ a database 
$\widetilde{\bm D} =  
\{\widetilde{R}^n,S^x,\widetilde{T}^n\}$
as follows.
Assume that $R^n$ is defined over the attribute $x$,
$S^x$ is defined over the attributes 
$(x,y)$, and $T^n$ is defined over the 
attribute $y$. Relation $S^x$ remains unchanged.
We transform relation $R^n$ into the relation $\widetilde{R}^n$ over the 
attributes $(z_1,x)$ for a new attribute $z_1$.
The relation $\widetilde{R}^n$ consists of the tuples 
$\{(z_1:d, x:a) | (x:a) \in R^n\}$, where 
$d$ is a fresh dummy value.
If a variable in $F_{Q,\bm D}$ is associated with the tuple
$(x:a)$ in $R^n$, we associate the same variable
with the tuple $(z_1:d, x:a)$ in $\widetilde{R}^n$.
Similarly, we transform relation $T^n$ into the relation
$\widetilde{T}^n$ over the attributes $(z_2,y)$ for a new attribute $z_2$.
The relation $\widetilde{T}^n$ consists of the tuples 
$\{(z_2:d,y:b)| (y:b) \in T^n\}$.  
If a variable in $F_{Q,\bm D}$ is associated with 
the tuple $(y:b)$ in $T^n$, we associate the same 
variable
with the tuple $(z_2:d, y:b)$ in 
$\widetilde{T}^n$.
Observe that $F_{Q,\bm D}$ is the lineage
of $\widetilde{Q}$ over $\widetilde{\bm D}$, which means that
$F_{Q,\bm D} \in \calC_{\widetilde{Q}}$.

\subsubsection{$C_{\widetilde{Q}} \subseteq C_{Q}$}
\label{sec:substituted_RST_lineage_in_RST_lineage}
This direction is analogous to the one shown 
in the previous section.
Consider the lineage $F_{\widetilde{Q},\widetilde{\bm D}} \in \calC_{\widetilde{Q}}$ for some database 
$\widetilde{\bm D} = \{\widetilde{R}^n,S^x,\widetilde{T}^n\}$. We show that 
$F_{\widetilde{Q},\widetilde{\bm D}} \in \calC_{Q}$.
We start with constructing a database  
$\bm D=\{R^n,S^x_{\text{new}},T^n\}$ 
from $\widetilde{\bm D}$.
Observe that in contrast to the construction 
in Section~\ref{sec:RST_lineage_in_substituted_RST_lineage}, 
we change also the relation $S^x$.
Assume that $\widetilde{R}^n$, $S^x$, and $\widetilde{T}^n$ in $\widetilde{\bm D}$
are defined over the attributes 
$(z_1,x)$, $(x,y)$, and 
respectively $(z_2,y)$. 
We denote the value domains of 
the attributes $z_1$, $x$, $z_2$, and $y$
by $\text{Dom}(z_1)$, $\text{Dom}(x)$, $\text{Dom}(z_2)$,
and respectively $\text{Dom}(y)$.
We construct from $\widetilde{R}^n$ the relation $R^n$ over the attribute 
$x'$ with domain $\text{Dom}(x') = 
\text{Dom}(z_1) \times \text{Dom}(x)$.
We define $R^n = \{(x':(a', a)) | (z_1:a', x:a) \in \widetilde{R}^n\}$.
If a variable in $F_{\widetilde{Q},\widetilde{\bm D}}$ is associated with the tuple 
$(z_1:a', x:a)$ in $\widetilde{R}^n$, we associate it with the
tuple $(x':(a', a))$ in $R^n$.
Analogously, we construct from $\widetilde{T}^n$ the relation 
$T^n$ over the attribute $y'$ with domain $\text{Dom}(y') = 
\text{Dom}(z_2) \times \text{Dom}(y)$.
We set $T^n = \{(y':(b', b)) | (z_2:b', y:b) \in \widetilde{T}^n\}$.
If a variable is associated with the  tuple 
$(z_2:b', y:b)$ in $\widetilde{T}^n$, we associate it with the
value $(y:(b', b))$ in $T^n$.
Finally, we construct from relation $S^x$
the relation $S^x_{\text{new}}$ over the attributes $(x',y')$ 
such that $S^x_{\text{new}} = \{(x':(a',a), y':(b',b))|
(z_1:a',x:a) \in \widetilde{R}^n, (x:a,y:b) \in S^x, \text{ and } (z_2:b',y:b)\in 
\widetilde{T}^n\}$.
Observe that $F_{\widetilde{Q},\widetilde{\bm D}}$ is the lineage
of $Q$ over $\bm D$. This means that 
$F_{\widetilde{Q},\widetilde{\bm D}} \in \calC_{Q}$.
\subsection{Proof of Lemma~\ref{lm:lineage-stretching}}
\textsc{Lemma} \ref{lm:lineage-stretching}.
\textit{$\widetilde{\calC_Q} \approx^P \calC_{\widetilde Q}$ holds for any CQ $Q$ and its stretching $\widetilde Q$.}

\medskip

The high-level idea of the bidirectional transformation is as follows:
Consider the lineage $F_{Q,D}$ of $Q$ 
over a database $\bm D$ 
and a variable $X$ associated with a tuple $\bm t = (\bm x:\bm a)$ in an endogenous relation $R^n$.
Assume that $\widetilde{F}_{Q,D}$ results from 
$F_{Q,D}$ by substituting $X$ 
with the disjunction $Z_1 \vee \cdots \vee Z_{\ell}$. 
Now, consider the database $\widetilde{\bm D}$ that results from 
$\bm D$ by stretching $R^n(\bm x)$ into $\widetilde{R}^n(z, \bm x)$ and 
replacing $\bm t = (\bm x:\bm a)$ with $\ell$ new tuples 
$\bm t_1 = (z:a_1, \bm x:\bm a), \ldots, \bm t_\ell = (z:a_{\ell}, \bm x:\bm a)$
where $a_1, \ldots, a_{\ell}$ are fresh values.
Then, $\widetilde{F}_{Q,D}$ is equivalent to the lineage $F_{\widetilde{Q},\widetilde{D}}$ of $\widetilde{Q}$ over $\widetilde{\bm D}$ and can be obtained from it in polynomial time (data complexity).

We now explain the transformations in more detail. 
Consider a CQ $Q$ and its stretching $\widetilde{Q}$.
In Section~\ref{sec:left_to_right} we show that $\calC_{\widetilde{Q}}~\lesssim^P~\widetilde{\calC_Q}$
and in Section~\ref{sec:right_to_left} we show that 
$\widetilde{\calC_Q}~\lesssim^P~\calC_{\widetilde{Q}}$.

\subsubsection{$\calC_{\widetilde{Q}}~\lesssim^P~\widetilde{\calC_Q}$}
\label{sec:left_to_right}
We describe a polynomial-time algorithm 
$A$ that transforms any function $F_{\widetilde{Q},\widetilde{\bm D}}\in\calC_{\widetilde{Q}}$ into an equivalent function from $\widetilde{\calC_Q}$, for some database $\widetilde{\bm D}$. The algorithm $A$ first constructs 
from $\widetilde{\bm D}$ a database $\bm D$, where the attributes added by stretching 
are discarded. Then, it transforms $F_{\widetilde{Q},\widetilde{\bm D}}$ into an equivalent function 
$\widetilde{F}$ in polynomial time such that 
$F_{Q,\bm D} \overset{OR}{\rightarrow} \widetilde{F}$, which means  
$\widetilde{F} \in \widetilde{\calC_Q}$. 
%
%
In the following, we first describe the construction of 
$\bm D$, then we give the definition of $\widetilde{F}$, and 
finally explain the transformation from 
$F_{\widetilde{Q},\bm D}$ into $\widetilde{F}$.

\paragraph{Construction of $\bm D$}
The exogenous relations in $\widetilde{\bm D}$ remain unchanged. 
The algorithm replaces  each endogenous relation $\widetilde{R}^n$ in $\widetilde{\bm D}$ with an endogenous relation 
$R^n$ constructed as follows. Let 
$(z, \bm y)$ be the attributes of 
$\widetilde{R}^n$ where $z$ is the attribute added due to stretching.
We set $R^n = \pi_{\bm y}\widetilde{R}$, i.e.,
$R^n$ is the projection of $\widetilde{R}$ onto $\bm y$.
Given a value tuple $\bm t$ over the variables $\bm y$, let 
$\bm Z$ be the set of variables associated to the tuples in $\widetilde{R}$
whose projection onto $\bm y$ is $\bm t$. The algorithm associates 
the fresh variable $X_{\bm Z}$ to the tuple $\bm t$ in $R$.
The construction time is linear in the size of $\widetilde{\bm D}$.


\paragraph{Definition of $\widetilde{F}$}
Let us denote the set of variables in 
$F_{Q,\bm D}$ by 
$\bm X$.
We define the substitution 
$\theta = \{X_{\bm Z}:= \bigvee_{Z \in {\bm Z}}Z | X_{\bm Z} \in {\bm X} \}$
and set $\widetilde{F} = F_{Q,\bm D}[\theta]$. 
It follows $F_{Q,\bm D} \overset{OR}{\rightarrow} \widetilde{F}$.

\paragraph{Transformation of $F_{\widetilde{Q},\widetilde{\bm D}}$ into $\widetilde{F}$}
The algorithm first constructs from $\bm D$ a database $\bm D'$ where 
each lineage variable $X_{\bm Z}$ 
is replaced by the disjunction $\bigvee_{Z \in \bm Z} Z$.
It then computes the lineage $F_{Q,\bm D'}$ of $Q$ over $\bm D'$. 
By construction,
it holds $F_{Q,\bm D'} = \widetilde{F}$ and $F_{Q,\bm D'} \equiv 
F_{\widetilde{Q},\widetilde{\bm D}}$.
The construction of $F_{Q,\bm D'}$
requires the computation of the join of the relations in $\bm D'$,
which can be done in time polynomial in the size of $\bm D'$ (hence, polynomial in the size 
of $\widetilde{\bm D}$) using any 
conventional join algorithm. 

\smallskip
We conclude that the overall transformation from
$F_{\widetilde{Q},\widetilde{\bm D}}$ into $\widetilde{F}$ takes time polynomial in the size of 
$F_{\widetilde{Q},\widetilde{\bm D}}$ and $\widetilde{\bm D}$.

\subsubsection{$\widetilde{\calC_Q}~\lesssim^P~\calC_{\widetilde{Q}}$}
\label{sec:right_to_left}
We give a polynomial-time algorithm 
$B$ that transforms functions in 
$\widetilde{\calC_Q}$
into equivalent functions  
in $\calC_{\widetilde{Q}}$.
Let $\widetilde{F} \in \widetilde{\calC_{Q}}$. 
This means that there is a database 
$\bm D$ and an OR-substitution $\theta$
such that $F_{Q,\bm D}[\theta] = \widetilde{F}$. 
We first explain how algorithm $B$ transforms
$\widetilde{F}$ in polynomial time into an equivalent function $\widetilde{F}'$ in DNF. 
Then, we show that $\widetilde{F}' \in \calC_{\widetilde{Q}}$, which 
concludes the proof. 

\paragraph{Transformation of $\widetilde{F}$ into $\widetilde{F}'$ in DNF}
Assume that $F_{Q,\bm D} = C_1 \vee \cdots \vee C_p$ 
where each $C_i$ is the conjunction of the variables in some set
$\bm X_i$. 
We set $\bm X = \bigcup_{i =1}^p \bm X_i$.
Assume that $\theta$ is defined as 
$\{X:= \bigvee_{Z \in \bm Z_X} Z | X \in \bm X\}$,
where for each $X \in \bm X$, $\bm Z_X$ is a set of
fresh variables. This means that 
$\widetilde{F} = C_1' \vee \cdots \vee C_p'$,
where 
$$C_i' = \bigwedge_{X \in \bm X_i}\ \ \bigvee_{Z \in \bm Z_X} Z.$$
Algorithm $B$ transform each such $C_i'$
 into a disjunction $C_i''$ of conjunctions.
 Assume that $\bm X_i = X_1 \wedge \cdots \wedge X_m$. 
 Then,
$$C_i'' = 
\bigvee_{Z_1 \in X_1, \ldots , Z_{m} \in X_m}\ \
\bigwedge_{j =1}^{m} Z_j.$$
The algorithms sets 
$\widetilde{F}' = C_1'' \vee \cdots \vee C_p''$.
The equivalence follows from the distributivity 
of $\vee$ over $\wedge$.
The transformation can be done in time  polynomial in the size of 
$\widetilde{F}$.

\paragraph{$\widetilde{F}' \in \calC_{\widetilde{Q}}$}
We turn $\bm D$ into a database 
$\widetilde{\bm D}$ such that the lineage of $\widetilde{Q}$
over $\widetilde{\bm D}$ is equal to $\widetilde{F}'$. 
The exogenous relations 
in $\bm D$ remain unchanged. 
For each endogenous relation $R^n$ over the attributes $\bm y$, 
we construct a relation $\widetilde{R}^n$ 
over $(z, \bm y)$.
For each value tuple $\bm t$ in $R^n$ associated with the lineage variable $X$, we add 
the following new tuples 
to $\widetilde{R}^n$. 
Let $\theta(X) = Z_1 \vee \cdots \vee Z_\ell$.  
We add to $\widetilde{R}^n$ the tuples $\bm t_1, \ldots, t_\ell$,
where each $\bm t_i$ results from $\bm t$ 
 by adding a fresh value for attribute $z$. 
We associate the tuples  $\bm t_1, \ldots, \bm t_\ell$
in $\widetilde{R}^n$ with the lineage variables $Z_1, \ldots, Z_\ell$,
respectively. It follows from the construction of $\widetilde{\bm D}$
that the lineage $F_{\widetilde{Q},\widetilde{\bm D}}$ is equal to 
$\widetilde{F}'$. 
Hence, $\widetilde{F}' \in \calC_{\widetilde{Q}}$.

\subsection{Proof of Lemma~\ref{lm:hierarchical-stretching}}
\textsc{Lemma}~\ref{lm:hierarchical-stretching}.
\textit{A CQ $Q$ is hierarchical iff its stretching $\widetilde Q$ is hierarchical.}

\medskip 

The main idea is that the class of hierarchical queries 
is closed under adding or removing  
variables that are contained in a single atom. To prove this formally, we first recall 
that for any two variables $x$ and $y$ in a hierarchical query, one of the following three properties (called {\em hierarchical properties} in the following) must hold:
$at(x)\cap at(y)=\emptyset$, $at(x)\subseteq at(y)$, or $at(y)\subseteq at(x)$. 
We show each direction of Lemma~\ref{lm:hierarchical-stretching} separately.

``$\Rightarrow$''-direction:
Assume that $Q$ is hierarchical. 
If $\widetilde{Q}$ does not contain any fresh variable,
then it is obviously hierarchical. So, let $x$ be a fresh variable
in $\widetilde{Q}$ that does not appear in $Q$ and let $y$ be 
an arbitrary variable in $\widetilde{Q}$.
By the definition of stretching, we have $|at(x)| =1$. 
For the sake of contradiction, assume that the hierarchical 
properties do not hold for $x$ and $y$ in $\widetilde{Q}$.
This means that $at(x)\cap at(y)\neq\emptyset$, yet $at(x) \not\subseteq at(y)$ and $at(y) \not\subseteq at(x)$. So $x$ appears in at least two atoms and the same for $y$. This contradicts our assumption that 
$|at(x)| =1$. Hence, $\widetilde{Q}$ must be hierarchical.


``$\Leftarrow$''-direction:
Assume that $\widetilde{Q}$ is hierarchical. 
Consider two distinct variables $x$ and $y$ in $Q$. By the definition of stretching, 
these variables 
must also appear in $\widetilde{Q}$. Since $\widetilde{Q}$ is hierarchical, 
one of the three hierarchical properties must hold for $x$ and $y$ in $\widetilde{Q}$. 
Since stretching only extends {\em existing} atoms, $at(x)$ and $at(y)$ are the same for both $Q$ and $\widetilde Q$. This means that one of the three hierarchical properties must also hold for $x$ and $y$ in $Q$. This implies that $Q$ is hierarchical.

\end{document}